\documentclass{article}

     \PassOptionsToPackage{numbers, compress}{natbib}

\usepackage[preprint]{neurips_2024}

\usepackage[utf8]{inputenc} 
\usepackage[T1]{fontenc}    
\usepackage[colorlinks]{hyperref}       
\hypersetup{
 linkcolor=sthlmRed
,citecolor=sthlmGreen
,urlcolor=sthlmBlue
}
\usepackage{url}            
\usepackage{booktabs}       
\usepackage{amsfonts}       
\usepackage{nicefrac}       
\usepackage{microtype}      
\usepackage{xcolor}         
\usepackage{amsmath}
\usepackage{amsfonts}
\usepackage{mathrsfs}
\usepackage{amssymb}
\usepackage{amsthm}
\usepackage{mathabx}
\usepackage{array,multirow,makecell}
\usepackage[mathcal]{euscript}
\usepackage[normalem]{ulem}
\usepackage{stmaryrd}
\usepackage{bbold}
\usepackage{paralist}
\usepackage{todonotes}
\usepackage{wrapfig}
\usepackage{subcaption}

\usepackage[capitalize,noabbrev]{cleveref}
\usepackage{enumerate}
\usepackage{algorithm}
\usepackage{algorithmic}

\newtheorem{theorem}{Theorem}[section]

\usepackage{thmtools}
\usepackage{thm-restate}
\usepackage{hyperref}
\usepackage{cleveref}

\usepackage{pgfplots}
\pgfplotsset{compat=1.7}
\usepackage{tikz}
\usetikzlibrary{shapes}
\usetikzlibrary{decorations.pathmorphing}
\usetikzlibrary{arrows}
\usetikzlibrary{automata}
\usetikzlibrary{fit}
\usetikzlibrary{shadows}
\usetikzlibrary{trees}
\usetikzlibrary{decorations}
\usetikzlibrary{decorations.text}
\usetikzlibrary{positioning}
\usetikzlibrary{patterns}
\usetikzlibrary{backgrounds}
\usetikzlibrary{matrix}
\usetikzlibrary{arrows.meta}
\usetikzlibrary{calc, shapes}
\usetikzlibrary {petri} 
\usetikzlibrary{chains,scopes}
\usetikzlibrary{spy}
\usepackage{array,multirow,makecell}
\usepackage[mathcal]{euscript}
\usepackage[normalem]{ulem}
\usepackage{stmaryrd}
\usepackage{bbold}
\usepackage{paralist}
\usepackage{sidecap}

\usepackage{xspace}
\def\ie{{\em i.e.,}\xspace}
\def\eg{{\em e.g.,}\xspace}
\def\cf{{\em cf.}\xspace}
\def\wrt{{\em w.r.t.}\xspace}

\DeclareMathOperator*{\argmax}{\arg\!\max}
\DeclareMathOperator*{\argmin}{\arg\!\min}

	\definecolor{sthlmLightBlue}{RGB}{214,237,252} 
		
	\definecolor{sthlmBlue}{RGB}{0,110,191} 

	\definecolor{sthlmLightGreen}{RGB}{213,247,244} 
		
	\definecolor{sthlmGreen}{RGB}{0,134,127} 

	\definecolor{sthlmLightGrey}{RGB}{213,217,225} 
		
	\definecolor{sthlmGrey}{RGB}{245,243,238} 
		
	\definecolor{sthlmDarkGrey}{RGB}{51,51,51} 

	\definecolor{sthlmLightOrange}{RGB}{255,215,210} 
		
	\definecolor{sthlmOrange}{RGB}{221,74,44} 

	\definecolor{sthlmLightPurple}{RGB}{241,230,252} 
		
	\definecolor{sthlmPurple}{RGB}{93,35,125} 

	\definecolor{sthlmLightRed}{RGB}{254,222,237} 
		
	\definecolor{sthlmRed}{RGB}{196,0,100} 

	\definecolor{sthlmYellow}{RGB}{252,191,10} 

\usepackage{colortbl}%
  
\usepackage{booktabs}

\newcommand{\pw}{{\textcolor{gray}{\circ}}}
\newcommand{\pb}{{\textcolor{gray}{\bullet}}}

\title{$\epsilon$-Optimally Solving Zero-Sum POSGs}

\author{%
  Erwan Escudie\\
  University of Groningen \\
  \texttt{e.c.escudie@rug.nl}\\
   \And
   Matthia Sabatelli \\
   University of Groningen \\
   \texttt{m.sabatelli@rug.nl} \\
   \And
   Jilles S. Dibangoye \\
   University of Groningen \\
   \texttt{j.s.dibangoye@rug.nl}
}

\begin{document}

\maketitle

\begin{abstract}
 A recent method for solving zero-sum partially observable stochastic games (zs-POSGs) embeds the original game into a new one called the occupancy Markov game. This reformulation allows applying \citeauthor{bellman}'s principle of optimality to solve zs-POSGs. However, improving a current solution requires solving a linear program with exponentially many potential constraints, which significantly restricts the scalability of this approach. This paper exploits the optimal value function's novel uniform continuity properties to overcome this limitation. We first construct a new operator that is computationally more efficient than the state-of-the-art update rules without compromising optimality. In particular, improving a current solution now involves a linear program with an exponential drop in constraints.  We then also show that point-based value iteration algorithms utilizing our findings improve the scalability of existing methods while maintaining guarantees in various domains.
 \end{abstract}

\section{Introduction}
The application of \citeauthor{bellman}'s principle of optimality to partially observable stochastic games (POSGs) can be traced back to the 70s when \citet{astro-1965:optimcontr:journal:174,smallwood,sondik78} introduced it for single-player POSGs---\ie partially observable Markov decision processes (POMDPs)---imperfect-information games against Nature. This approach embeds the original game into a fully observable Markov game, namely the {occupancy Markov game} (OMG): it solves the latter and then transfers its solution to the original game. Recently, a large body of works has successfully adapted this methodology to different many-player subclasses of POSGs ranging from common-reward  \citep{Szer05, nayyar2013decentralized,DibangoyeABC13,Dibangoye:OMDP:2016,oliehoek2013sufficient} to zero-sum games \citep{Wiggers16,7963513,horak2017heuristic,HorBos-aaai19,Buffet2020,DelBufDibSaf-DGAA-23,SokotaDL0KB23}, allowing the transfer of theories and algorithms from fully observable Markov games to POSGs, without compromising optimality. In common-reward settings, there has been significant progress in understanding the decomposition of the OMG into subgames, identifying uniform continuity properties of optimal solutions, and using them to solve subgames efficiently. In zero-sum settings, however, the scalability remains limited.

There are two distinct but interrelated reasons for the current ineffectiveness of this approach in zero-sum settings. The most widespread reason is the curse of dimensionality: the states of an OMG, \ie {occupancy states}, reside in a continuum whose dimensions grow exponentially with time. Uniform continuity properties of optimal solutions have been identified to overcome this drawback, but their effectiveness remains questionable \citep{Wiggers16,DelBufDibSaf-DGAA-23}. The less well-known reason for this poor scaling behavior is the inefficiency of the update operators used to improve a current solution. State-of-the-art update operators use linear programs with exponentially many potential constraints \citep{DelBufDibSaf-DGAA-23}. Although distinct, these limitations are interdependent. The weaker the uniform continuity properties designed to cope with the curse of dimensionality, the less effective the update operators are in improving current solutions. The conjunction of these two limitations has restricted the scaling behaviors of this approach in competitive settings.

This paper exploits recent uniform continuity properties to construct efficient update operators in zero-sum partially observable stochastic games. We build on the recent proof that the optimal value function of the corresponding OMG is the maximum of piecewise-linear and concave functions over occupancy states \citep{cunha2023convex}. Although this uniform continuity property is stronger than the previous ones, little is known about how it affects update operators. Our main contribution is the proof that the update operator exploiting this property consists of two components. Firstly, a linear program with an exponential drop in constraints enables any subgame to be efficiently solved. Secondly, a polynomial-time update rule uses the solution of a subgame to improve the current value function. We also show how to use this update operator in the point-based value iteration algorithm. Experimental results conducted on a set of benchmarks from the literature support the effectiveness of our approach.

\section{Background}
This section presents the original formulation of zero-sum partially observable stochastic games alongside its reformulation.

\subsection{Problem Formulation}

A zero-sum, partially observable stochastic game (zs-POSGs) \citep{HansenBZ04} is a tuple $M= ( X,U,Z,p,r,s_0,\gamma, \ell)$ where $\pb$ and $\circ$ represent the maximizing and minimizing players, respectively. $X$ denotes a finite set of hidden states, denoted $x$ or $y$. $U = U^{\pb}\times U^{\pw}$ represents the finite set of joint controls $u = (u^{\pb}, u^{\pw})$. Similarly, $Z = Z^{\pb}\times Z^{\pw}$ denotes the finite set of joint observations $z=(z^{\pb}, z^{\pw})$. The transition function $p\colon X\times U\to \triangle(X\times Z)$ specifies the probability $p_{xy}^{uz} = p(y,z|x,u)$ of the game being in state $y$ upon taking control $u$ in state $x$ and receiving observation $z$. The reward function $r\colon X\times U\to \mathbb{R}$ defines the immediate reward $r_{xu} = r(x,u)$ that player ${\pb}$ receives upon taking control $u$ in state $x$. Additionally, $s_0$ represents the initial state distribution, $\gamma\in [0,1)$ is a discount factor, and $\ell<\infty$ is the number of game stages.

The primary aim of this study is to tackle general zs-POSGs. However, this research also explores the application of our findings in subclasses. The subclasses under scrutiny in this paper are as follows: i) $M$ with full observability, see \citep{filar-competitiveMDP}, which implies the existence of a function, $f_{\mathtt{fo}}\colon Z^{\pb}\cup Z^{\pw} \to X$, such that for every non-zero value of $p(z,y|x,u)$, $f_{\mathtt{fo}}(z^{\pb})=f_{\mathtt{fo}}(z^{\pw})=y$; ii) $M$ with public actions and observations, see \citep{ghosh2004zero}, which implies the existence of a function, $f_{\mathtt{pao}}\colon Z^{\pb}\cup Z^{\pw} \to Z\times U$, such that for every non-zero value of $p(z,y|x,u)$,  $f_{\mathtt{pao}}(z^{\pb})=f_{\mathtt{pao}}(z^{\pw})=(z,u)$.

Optimally solving $M$ aims at finding a joint policy $a_{0:} = (a^{\pb}_{0:},a^{\pw}_{0:})$, one policy for each player. A policy of player ${\pb}$ (resp. ${\circ}$), denoted  $a^{\pb}_{0:} = (a^{\pb}_0,\ldots, a^{\pb}_{\ell-1})$, is a sequence of decision rules from stage $0$ down to $\ell-1$, one  decision rule per stage. A $\tau$-stage decision rule $a^{\pb}_\tau\colon O^{\pb}_\tau \to \triangle(U^{\pb})$ maps histories $o^{\pb}_\tau = (u^{\pb}_{0:\tau-1},z^{\pb}_{1:\tau})$ of controls and observations of player ${\pb}$ to distributions over its controls. A $\tau$-stage joint decision rule $a_\tau\colon O_\tau\to\triangle(U)$ maps joint histories $o_\tau = (o^{\pb}_\tau,o^{\pw}_\tau)$ of actions and observations to distributions over joint controls, \ie for all joint control $u\in U$, joint history $o\in O_\tau$, $a_\tau(u|o) = a^{\pb}_\tau(u^{\pb}_\tau|o^{\pb}_\tau)\cdot a^{\pw}_\tau(u^{\pw}_\tau|o^{\pw}_\tau)$. We will denote $A^{\pb}_\tau$ (resp. $A^{\pw}_\tau$) the set of all $\tau$-stage  individual decision rules $a^{\pb}_\tau$ (resp. $a^{\pw}_\tau$), and $A_\tau = A^{\pb}_\tau \times A^{\pw}_\tau$ the set of all $\tau$-stage joint decision rules.

The state- and action-value functions under a fixed joint policy satisfy  \citeauthor{bellman} equations: for any stage $\tau$, joint policy $a_{\tau:} = (a_\tau,a_{\tau+1},\ldots,a_{\ell-1})$,  $\alpha^{a_{\tau:}}_\tau\colon (x,o)\mapsto \textstyle \sum_{u}a_\tau(u|o) \cdot \beta^{a_{\tau:}}_\tau(x,o,u)$ and 
$\beta^{a_{\tau:}}_\tau \colon(x,o,u)\mapsto\textstyle r(x,u) +\gamma\sum_{y}\sum_{z}\alpha^{a_{\tau+1:}}_{\tau+1}(y,(o,u,z))$
with boundary condition $\alpha^{\cdot}_\ell(\cdot) = \beta^{\cdot}_\ell(\cdot) = 0$.  An optimal joint policy $\bar{a}_{0:} = (\bar{a}^{\pb}_{0:},\bar{a}^{\pw}_{0:})$, \ie a Nash equilibrium, consists of player policies, \eg $\bar{a}^{\pb}_{0:}$, whose worst-case expected returns, \eg $\upsilon^{\pb}_0(s_0;\bar{a}^{\pb}_{0:}) =\textstyle \min_{a^{\pw}_{0:}}~\sum_{x} s_0(x)\cdot \alpha^{ \bar{a}^{\pb}_{0:},a^{\pw}_{0:} }_0(x)$, regardless the opponent policy are the value of the game $\upsilon^*_0(s_0)=\textstyle \max_{a^{\pb}_{0:}} \min_{a^{\pw}_{0:}}~\sum_{x} s_0(x)\cdot \alpha^{ a^{\pb}_{0:},a^{\pw}_{0:} }_0(x)$.

Optimally solving $M$ using \citeauthor{bellman}'s principle of optimality is perceived as challenging due to the complexity involved in defining a suitable common ground that allows the original game to be segmented into subgames that can be solved recursively  \citep{HansenBZ04}. To comprehend this challenge better, it is imperative to note that each player in the game acts simultaneously without having the ability to perceive the current state of the game or communicate their actions and observations with others. Consequently, neither the individual players' histories nor the joint histories can act as a common ground. This insight explains the impetus for a problem reformulation.

\subsection{Problem Reformulation}

To overcome the limitation of the original formulation, \citet{Szer05} and later \citet{nayyar2013decentralized,DibangoyeABC13,Dibangoye:OMDP:2016,oliehoek2013sufficient} suggest formalizing $M$ from the perspective of an offline central planner. This planner acts on behalf of both players by selecting a joint decision rule to be executed at each stage based on all data available about the game at the planning phase, namely the information state. The information state at the end of stage $\tau$, denoted $\iota_{\tau+1} = (s_0,a_{0:\tau})$, is a sequence of joint decision rules the planner selected starting at the initial state distribution $s_0$. Hence, the information state satisfies the following recursion: $\iota_0 = s_0$ and $\iota_{\tau+1} = (\iota_\tau, a_\tau)$ for all stage $\tau$, resulting in an ever-growing sequence. \citet{DibangoyeABC13,Dibangoye:OMDP:2016} introduce the concept of occupancy state along with bookkeeping to replace information state without losing the ability to find an optimal joint policy. The $\tau$-stage occupancy state $s_\tau$ is a distribution over hidden states and joint histories conditional on information state $\iota_\tau$ at stage $\tau$, \ie  $s_\tau\colon (x,o) \mapsto \Pr\{x,o|\iota_\tau\}$. It is important to stress that in general the occupancy state is neither public nor accessible to either player at the online execution phase. It summarizes the total information available to a central algorithm to solve the original game using \citeauthor{bellman}'s principle of optimality at an offline planning phase.  

It is worth noting that occupancy states possess several crucial properties. When restricted to subclasses, occupancy states become more compact. For instance, when a game exhibits full observability, the occupancy state represents a state. If the game has either public actions and observations or one-sided partial observability, the occupancy state becomes a belief state \citep{HorBos-aaai19}. Lastly, for one-sided information-sharing games, the occupancy state is a distribution over belief states \citep{pmlr-v119-xie20a}. Occupancy states serve as sufficient statistics of the information state when estimating the immediate reward to be gained by executing a joint decision rule:
$
R\colon (s_\tau,a_\tau) \mapsto  \sum_{x\in X}\sum_{o\in O_\tau} s_\tau(x,o)\sum_{u\in U} a_\tau(u|o) \cdot r(x,u).
$
In addition, the $\tau$-stage occupancy state $s_\tau$ is a sufficient statistic of the information state $\iota_\tau$ to predict the next occupancy state $s_{\tau+1}$ upon taking a joint decision rule $a_\tau$, \ie 
$s_{\tau+1}\colon (y, (o,u,z)) \mapsto a_\tau(u|o) \sum_{x\in X} s_\tau(x,o)\cdot p(y,z|x,u)$,
where $T\colon (s_\tau,a_\tau) \mapsto s_{\tau+1}$ describes the transition rule. Together, these properties make the occupancy state a sufficient statistic of the information state when estimating the value function under a fixed joint policy. Unfortunately, optimal joint policies cannot depend on occupancy states since players cannot access them. This insight makes the central planning process an open-loop control problem. However, a bookkeeping strategy enables us to keep track of relevant information to ensure that the planning process stores the optimal joint policy \citep{DelBufDibSaf-DGAA-23}. 

The open-loop control problem $M'= (S,A,T,R,s_0,\gamma, \ell)$ that occupancy states describe is called the occupancy Markov game. Similarly to POMDPs, it was proven that $M$ can be recast into $M'$ and an optimal solution of $M'$ is also an optimal solution for $M$ \citep{Wiggers16,DelBufDibSaf-DGAA-23}. $M'$ is a non-observable deterministic Markov game with respect to $M$, where state space $S = \cup_{\tau=0}^{\ell-1}~ S_\tau$ is the set of occupancy states up to stage $\ell-1$; action space $A = \cup_{\tau=0}^{\ell-1}~ A_\tau$ is the set of joint decision rules up to stage $\ell-1$; the transition rule is $T$; the reward model is $R$; and quantities $s_0, \gamma,$ and $\ell$ are as in the original game $M$. 
Optimally solving $M'$ (resp. $M$) aims at finding a joint policy $ (\bar{a}^{\pb}_{0:},\bar{a}^{\pw}_{0:})$ such that $\upsilon^{\pb}_0(s_0;\bar{a}^{\pb}_{0:})=\upsilon^{\pw}_0(s_0;\bar{a}^{\pw}_{0:})=\upsilon^*_0(s_0)$. The application of \citeauthor{bellman}'s principle of optimality allows us to compute $\upsilon^*_0(s_0)$ by solving subgames recursively, \ie for any stage $\tau$, occupancy state $s_\tau$, 
\begin{align}
\upsilon_\tau^*(s_\tau) &=\textstyle \max_{a^{\pb}_\tau\in A^{\pb}_\tau}\min_{a^{\pw}_\tau\in A^{\pw}_\tau} R(s_\tau,a_\tau) + \gamma \upsilon^*_{\tau+1}(T(s_\tau,a_\tau))
\label{eqn:bellman:optimality}
\end{align}
with boundary condition $\upsilon_\ell^*(\cdot)=0$. As we solve the problem $M'$, we need to ensure consistency with $M$. For this, we use a bookkeeping strategy to keep track of partial policies $\bar{a}^{\pb}_{\tau:}$ and $\bar{a}^{\pw}_{\tau:}$. These policies ensure that we always meet constraint $\upsilon_\tau^{\pw}(s_\tau;\bar{a}^{\pw}_{\tau:})=\upsilon_\tau^{\pb}(s_\tau;\bar{a}^{\pb}_{\tau:}) = \upsilon_\tau^*(s_\tau)$ throughout the planning process. To simplify notation, we use $\upsilon_\tau^{\pw}(s_\tau)$ and $\upsilon_\tau^{\pb}(s_\tau)$ instead of $\upsilon_\tau^{\pw}(s_\tau;\bar{a}^{\pw}_{\tau:})$ and $\upsilon_\tau^{\pb}(s_\tau;\bar{a}^{\pb}_{\tau:})$ when the policies are not explicitly needed.

In principle, exact backward induction should apply to $M'$. Unfortunately, the occupancy states lie in a continuum, which makes exact backward induction infeasible. When the optimal value functions exhibit uniform continuity properties, one can restrict attention to a small set of representative occupancy states and iteratively apply value updates to those points while preserving the ability to achieve the $\epsilon$-optimal value of the game. Let $\upsilon^{\pb}_{\tau+1}$ and $\upsilon^{\pw}_{\tau+1}$ be the value functions at stage $\tau+1$ of players ${\pb}$ and ${\circ}$, respectively. Operators at any point $s_\tau$ involved in solving $M'$ include the greedy-action selection operators $\mathbb{G}^{\pb} \colon (\upsilon^{\pb}_{\tau+1}, s_\tau)\mapsto a^{\pb}_\tau$ and $\mathbb{G}^{\pw} \colon (\upsilon^{\pw}_{\tau+1}, s_\tau)\mapsto a^{\pw}_\tau$, and the \citeauthor{bellman}'s update operators $\mathbb{H}^{\pb}\colon (\upsilon^{\pb}_{\tau+1}, s_\tau,a_\tau) \mapsto \upsilon^{\pb}_\tau$ and $\mathbb{H}^{\pw}\colon (\upsilon^{\pw}_{\tau+1}, s_\tau,a_\tau) \mapsto \upsilon^{\pw}_\tau$.

\subsection{Uniform Continuity Properties}

Various uniform continuity properties of optimal value functions have been identified recently to define efficient point-based operators  $\mathbb{G}$ and $\mathbb{H}$ \citep{Wiggers16,DelBufDibSaf-DGAA-23,cunha2023convex}. To discuss these properties, we need two concepts associated with an occupancy state, \ie marginal and conditional occupancy states. For any occupancy state $s_\tau$, the marginal occupancy state $s^{\mathtt{m},{\circ}}_\tau$ of player ${\circ}$ is defined as the marginal distribution of $s_\tau$ over histories $O_\tau^{\pw}$, \ie for all  history $o^{\pw} \in O_\tau^{\pw}$, $s^{\mathtt{m},{\circ}}_\tau(o^{\pw}) = \sum_x\sum_{o^{\pb}} s_\tau(x,(o^{\pw},o^{\pb}))$. Furthermore, for any occupancy state $s_\tau$ and any history $o^{\pw}\in O_\tau^{\pw}$, the conditional occupancy state $s^{\mathtt{c},o^{\pw}}_\tau$ of $s_\tau$ associated with history $o^{\pw}$ is defined as the conditional distribution of $s_\tau$ associated with history $o^{\pw}$, \ie for all $x\in X$ and $o^{\pb} \in O_\tau^{\pb}$, $s_\tau(x,(o^{\pw},o^{\pb})) = s^{\mathtt{c},o^{\pw}}_\tau(x,o^{\pb} ) \cdot s^{\mathtt{m},{\circ}}_\tau(o^{\pw})$. We shall use $s^{\mathtt{c},{\circ}}_\tau$ to denote a family $\{s^{\mathtt{c},o^{\pw}}_\tau | o^{\pw}\in O_\tau^{\pw}\}$  of conditional occupancy states and $s^{\mathtt{c},{\circ}}_\tau\odot s^{\mathtt{m},{\circ}}_\tau$ to describe occupancy state $s_\tau$ such that, for all $x\in X$ and $(o^{\pw},o^{\pb})\in O_\tau$, $s_\tau(x,(o^{\pw},o^{\pb})) = s^{\mathtt{c},o^{\pw}}_\tau(x,o^{\pb} ) \cdot s^{\mathtt{m},{\circ}}_\tau(o^{\pw})$. We are now ready to state the main uniform continuity properties.

\begin{figure}[H]
\centering
\begin{tikzpicture}[
        scale=1.3,
        IS/.style={sthlmRed, thick},
        LM/.style={sthlmRed, thick},
        axis/.style={very thick, ->, >=stealth', line join=miter},
        important line/.style={thick}, dashed line/.style={dashed, thin},
        every node/.style={color=black},
        dot/.style={circle,fill=red,minimum size=8pt,inner sep=0pt, outer sep=-1pt},
    ]

    \draw[axis, -, line width=.3em] (-6,2.75) -- (-6,0) -- (-2.25,0) -- (-2.25,2.75);
    \node at (-6.2,-0.3) {\textbf{A}};
    \node at (-0.2,-0.3) {\textbf{B}};
    
    \draw[-, sthlmRed, line width=.1em] (-2.25,2)--(-6,0);

    \draw[-, sthlmGreen, line width=.1em] (-6,1.65) -- (-2.25,1.25);     

    \draw[-, sthlmBlue, line width=.1em] (-6,2.25)--(-2.25,0);
 
    \node[scale=1, sthlmBlue] at (-2.75,-0.3) {$\pmb{s^{\mathtt{m},\circ}}$};
    \node[scale=1, sthlmRed] at (-5.07,-0.3) {$\pmb{s^{\mathtt{m},\circ}}$};

    \draw[-, dotted, sthlmOrange] (-3.67,-0.1) -- (-3.67,2.75);    
    \node[scale=1, sthlmOrange] at (-3.67,-0.3) {$\pmb{s^{\mathtt{m},\circ}}$};

    \draw[-, dotted, sthlmRed] (-5.07,2.75) -- (-5.07,-0.1);
    \node[scale=.5, sthlmRed] at (-6.5,.5) {$g_{s^{\mathtt{c},\circ}}(\pmb{s^{\mathtt{m},\circ}})$};    
    \draw[-, dotted, sthlmRed] (-6.1,.5) -- (-2.25,.5);    
    \node[rotate=90, scale=2, sthlmRed] at (-5.07,.5) {$\bullet$};
    \node[rotate=90, scale=1, sthlmRed] at (-6,.5) {$\bullet$};
    \node[rotate=90, scale=1, sthlmRed] at (-5.07,0) {$\bullet$};

    \draw[-, dotted, sthlmBlue] (-2.75,-0.1) -- (-2.75,2.75);    
    \node[scale=.5, sthlmBlue] at (-6.5,.3) {$g_{s^{\mathtt{c},\circ}}(\pmb{s^{\mathtt{m},\circ}})$};    
    \draw[-, dotted, sthlmBlue] (-6.1,.3) -- (-2.25,.3);    
    \node[scale=2, sthlmBlue] at (-2.75,.3) {$\bullet$};
    \node[scale=1, sthlmBlue] at (-6,.3) {$\bullet$};
    \node[scale=1, sthlmBlue] at (-2.75,0) {$\bullet$};

    \draw[axis, -, line width=.3em] (0,2.75) -- (0,0) -- (3.75,0) -- (3.75,2.75);

    \draw[-, sthlmRed, line width=.1em] (3.75,2)--(0,0);

    \draw[-, sthlmGreen, line width=.1em] (0,1.65) -- (3.75,1.25);     

    \draw[-, sthlmBlue, line width=.1em] (0,2.25)--(3.75,0);

    \draw[-, dotted, sthlmOrange] (2.33,-0.1) -- (2.33,2.75);    
    \node[scale=1, sthlmOrange] at (2.33,-0.3) {$\pmb{s^{\mathtt{m},\circ}}$};
    
    \node[scale=.5, sthlmOrange] at (-.5,1.5) {$\bar{\upsilon}_\tau(\pmb{s^{\mathtt{m},\circ}  \odot  s^{\mathtt{c},\circ}})$};    

    \node[scale=.5] at (.75,1.18) {$\pmb{\kappa\|\textcolor{sthlmOrange}{s}-\textcolor{sthlmOrange}{s^{\mathtt{m},\circ}}  \odot  \textcolor{sthlmBlue}{s^{\mathtt{c},\circ}}\|_1}$};   
    
    \node[scale=.5, sthlmBlue] at (-.5,.85) {$g_{s^{\mathtt{c},\circ}}(\pmb{\textcolor{sthlmOrange}{s^{\mathtt{m},\circ}}})$};    
    \draw[-, dotted, sthlmOrange] (-.1,.85) -- (3.75,.85);    
    \node[rotate=90, scale=2, sthlmOrange] at (2.33,.85) {$\bullet$};
    \node[rotate=90, scale=1, sthlmOrange] at (2.33,0) {$\bullet$};
    \node[rotate=90, scale=1, sthlmOrange] at (0,.85) {$\bullet$};
    \node[rotate=90, scale=1, sthlmOrange] at (0,1.5) {$\bullet$};
    \draw[dotted, sthlmOrange] (0,.85) .. controls (.25,.85) and (.25,1.5) .. (0,1.5);
\end{tikzpicture}
\caption{Generalization across marginal occupancy states of the value function given by a collection $V = \{\textcolor{sthlmBlue}{g_{s^{\mathtt{c},\circ}}}, \textcolor{sthlmRed}{g_{s^{\mathtt{c},\circ}}}, \textcolor{sthlmGreen}{g_{s^{\mathtt{c},\circ}}}\}$ of linear functions over unknown marginal occupancy states. Figure \textbf{A} shows no generalization on marginal occupancy state $\textcolor{sthlmOrange}{\pmb{s^{\mathtt{m},\circ}}}$ because $\textcolor{sthlmOrange}{\pmb{s^{\mathtt{m},\circ}}} \notin \{\textcolor{sthlmBlue}{s^{\mathtt{m},\circ}}, \textcolor{sthlmRed}{s^{\mathtt{m},\circ}}, \textcolor{sthlmGreen}{s^{\mathtt{m},\circ}}\}$, \cf Theorem \ref{thm:wiggers}. Figure \textbf{B} shows generalization over unknown marginal occupancy state $\textcolor{sthlmOrange}{\pmb{s^{\mathtt{m},\circ}}}$ from known marginal occupancy state $\textcolor{sthlmBlue}{\pmb{s^{\mathtt{m},\circ}}}$ with offset $\kappa\|\textcolor{sthlmOrange}{s}-\textcolor{sthlmOrange}{s^{\mathtt{m},\circ}}  \odot  \textcolor{sthlmBlue}{s^{\mathtt{c},\circ}}\|_1$, \cf Theorem \ref{thm:delage}. \textbf{Best viewed in color}.}
\label{figure:weak:convex:function:representations}
\end{figure} 

\begin{theorem}[Adapted from \citet{Wiggers16}]
\label{thm:wiggers}
For any arbitrary $M'$, the optimal value functions $(\upsilon^*_0,\ldots, \upsilon^*_\ell)$ solutions of (\ref{eqn:bellman:optimality}) are convex over marginal occupancy states for a fixed conditional occupancy-state family, \ie there exists collections $(G_0,\ldots,G_\ell)$ of linear functions over marginal occupancy states such that: for any stage $\tau$, occupancy state $s_\tau = s^{\mathtt{c},\circ}_\tau\odot s^{\mathtt{m},\circ}_\tau$, 
\begin{align*} 
\upsilon^*_\tau(s_\tau) &=\textstyle \max_{g_{s^{\mathtt{c},\circ}_\tau} \in G_\tau} g_{s^{\mathtt{c},\circ}_\tau} (s^{\mathtt{m},\circ}_\tau),
\end{align*}
where $g_{s^{\mathtt{c},\circ}_\tau}\colon O_\tau^{\pw} \to \mathbb{R}$ is a function  associated with conditional occupancy-state family $s^{\mathtt{c},\circ}_\tau$.
\end{theorem}
\citet{Wiggers16} presents a detailed proof of Theorem \ref{thm:wiggers}, which asserts that if two occupancy states exhibit identical conditional occupancy-state families, it is possible to generalize the value from the first occupancy state to the second one. Unfortunately, this conditional uniform continuity property alone does not facilitate value generalization across different occupancy states. Figure \ref{figure:weak:convex:function:representations} (\textbf{A}) illustrates the lack of generalization capabilities across occupancy states.  To overcome this limitation, \citet{DelBufDibSaf-DGAA-23} blends the Lipschitz continuity and the conditional uniform continuity property, hence allowing for a value generalization across unknown occupancy states.

\begin{theorem}[Adapted from \citet{DelBufDibSaf-DGAA-23}]
\label{thm:delage}
For any arbitrary $M'$, the optimal value functions $(\upsilon^*_0,\ldots, \upsilon^*_\ell)$ solutions of (\ref{eqn:bellman:optimality}) are Lipschitz continuous over occupancy states, \ie there exists collections $(G_0,\ldots,G_\ell)$ of linear functions over marginal occupancy states such that: for any stage $\tau$,  $\kappa_\tau$ is the Lipschitz constant  associated with $\upsilon^*_\tau$,  and occupancy state $s_\tau = \bar{s}^{\mathtt{c},\circ}_\tau\odot s^{\mathtt{m},\circ}_\tau$, 
\begin{align*} 
\upsilon^*_\tau(s_\tau) &\leq\textstyle  g_{s^{\mathtt{c},\circ}_\tau} (s^{\mathtt{m},\circ}_\tau) + \kappa_\tau \|s_\tau - s^{\mathtt{c},\circ}_\tau\odot s^{\mathtt{m},\circ}_\tau \|_1
\end{align*}
where $g_{s^{\mathtt{c},\circ}_\tau}\colon O_\tau^{\pw} \to \mathbb{R}$ is any function in $V_\tau$ associated with conditional occupancy-state family $s^{\mathtt{c},\circ}_\tau$ .
\end{theorem}
The process of generalizing values across different occupancy states is hindered by imprecise approximations from using Lipschitz constants in Theorem \ref{thm:delage}, see also  Figure \ref{figure:weak:convex:function:representations} (\textbf{B}). In addition, greedy-action selection operators, \eg $\mathbb{G}^{\pb} \colon (\upsilon^{\pb}_{\tau+1}, s_\tau)\mapsto \argmax_{a^{\pb}_\tau} \min_{a^{\pw}_\tau} R(s_\tau,a_\tau) + \gamma\upsilon^{\pb}_{\tau+1}(T(s_\tau,a_\tau))$, using this non-linear value function approximation require the enumeration of exponentially many possible decision rules $a^{\pw}_\tau$. \citet{DelBufDibSaf-DGAA-23} performs a greedy-action selection using a linear program with exponentially many potential constraints. For instance, if $\upsilon^{\pb}_{\tau+1}$ is known, the following linear program, \ie $\max \{v | a^{\pb}_\tau\in A^{\pb}_\tau, v \in \mathbb{R}, v \leq  R(s_\tau,a_\tau) + \gamma\upsilon^{\pb}_{\tau+1}(T(s_\tau,a_\tau)),~ \forall a^{\pw}_\tau\in A^{\pw}_\tau \}$,  performs a greedy-action selection with $\pmb{O}(|U^{\pw}|^{|\bar{O}^{\pw}_\tau|})$ constraints where $\bar{O}^{\pw}_\tau = \{o^{\pw}_\tau| o^{\pw}_\tau\in O^{\pw}_\tau, \Pr\{o^{\pw}_\tau|s_\tau\} > 0\}$. \citet{DelBufDibSaf-DGAA-23} mitigated this burden by considering only previously experienced (stochastic) decision rules $a^{\pw}_\tau$ instead of all of them. These drawbacks nonetheless impede algorithmic efficiency and call for alternative approaches to improve efficiency.

\begin{theorem}[Adapted from \citet{cunha2023convex}]
\label{thm:cunha}
For any arbitrary $M'$, the optimal value functions $(\upsilon^*_0,\ldots, \upsilon^*_\ell)$ solutions of (\ref{eqn:bellman:optimality}) are maximum of piece-wise linear and concave functions of occupancy states, \ie there exists families  $(\mathbb{V}^*_0,\ldots,\mathbb{V}^*_\ell)$ of collections of linear functions over conditional occupancy states such that: for any stage $\tau$ 
\begin{align} 
\upsilon^*_\tau\colon s_\tau &\mapsto \max_{V_\tau\in \mathbb{V}^*_\tau} \sum_{o^{\pw}\in O^{\pw}}s^{\mathtt{m},\circ}_\tau(o^{\pw}) \min_{\alpha_\tau\in V_\tau} \alpha_\tau(s^{\mathtt{c},o^{\pw}}_\tau)
\label{eqn:thm:greedy}
\end{align}
where collection $V_\tau$ represents the optimal (piecewise-linear and concave) value function of player ${\circ}$ given that player ${\pb}$ acts according to a fixed policy.
\end{theorem}
\begin{figure}[H]
\centering
\begin{tikzpicture}[
        scale=1.3,
        IS/.style={sthlmRed, thick},
        LM/.style={sthlmRed, thick},
        axis/.style={very thick, ->, >=stealth', line join=miter},
        important line/.style={thick}, dashed line/.style={dashed, thin},
        every node/.style={color=black},
        dot/.style={circle,fill=red,minimum size=8pt,inner sep=0pt, outer sep=-1pt},
    ]

    \draw[axis, -, line width=.3em] (-6,2.75) -- (-6,0) -- (-2.25,0) -- (-2.25,2.75);
    \node at (-6.2,-0.3) {\textbf{A}};
    \node at (-0.2,-0.3) {\textbf{B}};

    \draw[-, sthlmRed, line width=.1em] (-6,2) -- (-5.07, 2.5) -- (-3.5, 2.15) -- (-2.25, 1.25);

    \draw[-, sthlmGreen, line width=.1em] (-6,1.65) -- (-4.6, 2.75) -- (-3.25, 2.35) -- (-2.25,0);

    \draw[-, sthlmBlue, line width=.1em] (-6,0)-- (-5.5,1.5) -- (-4,2.65) -- (-2.75,2.35) -- (-2.25, .65);
 
    \node[scale=1, sthlmYellow] at (-2.75,-0.3) {$\pmb{s^{\mathtt{c},o^{\circ}}}$};
    
    \node[scale=1, sthlmPurple] at (-5.07,-0.3) {$\pmb{s^{\mathtt{c},o^{\circ}}}$};

    \node[scale=1, sthlmOrange] at (-3.67,-0.3) {$\pmb{s^{\mathtt{c},o^{\circ}}}$};

    \draw[-, dotted, sthlmYellow] (-2.75,2.75) -- (-2.75,0);
    \draw[-, dotted, sthlmYellow] (-6.1,2.35) -- (-2.25,2.35);    
    \node[rotate=90, scale=2, sthlmYellow] at (-2.75,2.35) {$\bullet$};
    \node[scale=.5, sthlmBlue] at (-6.5,2.35) {$\omega(\textcolor{sthlmYellow}{\pmb{s^{\mathtt{c},o^{\circ}}}})$};    
    \node[scale=1, sthlmYellow] at (-2.75,0) {$\bullet$};
    \node[scale=1, sthlmYellow] at (-6,2.35) {$\bullet$};

    \draw[-, dotted, sthlmOrange] (-3.67,2.75) -- (-3.67,0);
    \draw[-, dotted, sthlmOrange] (-6.1,2.58) -- (-2.25,2.58);    
    \node[rotate=90, scale=2, sthlmOrange] at (-3.67,2.58) {$\bullet$};
    \node[scale=.5, sthlmBlue] at (-6.5,2.58) {$\omega(\textcolor{sthlmOrange}{\pmb{s^{\mathtt{c},o^{\circ}}}})$};    
    \node[scale=1, sthlmOrange] at (-3.67,0) {$\bullet$};
    \node[scale=1, sthlmOrange] at (-6,2.58) {$\bullet$};

    \draw[-, dotted, sthlmPurple] (-5.07,2.75) -- (-5.07,0);
    \draw[-, dotted, sthlmPurple] (-6.1,1.83) -- (-2.25,1.83);    
    \node[rotate=90, scale=2, sthlmPurple] at (-5.07,1.83) {$\bullet$};
    \node[scale=.5, sthlmBlue] at (-6.5,1.83) {$\omega(\textcolor{sthlmPurple}{\pmb{s^{\mathtt{c},o^{\circ}}}})$};    
    \node[scale=1, sthlmPurple] at (-5.07,0) {$\bullet$};
    \node[scale=1, sthlmPurple] at (-6,1.83) {$\bullet$};

    \draw[axis, -, line width=.3em] (0,2.75) -- (0,0) -- (3.75,0) -- (3.75,2.75);

    \draw[-, sthlmGreen, line width=.1em] (3.75,2) -- (2.55,1.36015);
    \draw[-, sthlmGreen, line width=.1em, dashed] (2.55,1.36015) -- (1.2,0.6406);
    \draw[-, sthlmGreen, line width=.1em, dashed] (1.2,0.6406) -- (0,0);

    \draw[-, sthlmBlue, line width=.1em, dashed] (0,1.65)--(1.2,1.5216);     
    \draw[-, sthlmBlue, line width=.1em] (1.2,1.5216)--(2.55,1.37715);     
    \draw[-, sthlmBlue, line width=.1em, dashed] (2.55,1.37715) -- (3.75,1.25);     

     \draw[-, sthlmRed, line width=.1em] (0,2.25)--(1.2,1.53); 
    \draw[-, sthlmRed, line width=.1em, dashed] (1.2,1.53)--(2.55,0.72); 
    \draw[-, sthlmRed, line width=.1em, dashed] (2.55,0.72)--(3.75,0);

    \draw[-, dotted, sthlmBlue] (2.33,-0.1) -- (2.33,2.75);    
    \draw[-, dotted, sthlmBlue] (-.1,1.42) -- (3.75,1.42);    
    \node[scale=1, sthlmBlue] at (2.33,-0.3) {$\pmb{s}$};    
    \node[scale=.5, sthlmBlue] at (-.35,1.42) {$\omega(\textcolor{sthlmBlue}{\pmb{s}})$};    
    \node[rotate=90, scale=2, sthlmBlue] at (2.33,1.42) {$\bullet$};
    \node[rotate=90, scale=1, sthlmBlue] at (2.33,0) {$\bullet$};
    \node[rotate=90, scale=1, sthlmBlue] at (0,1.42) {$\bullet$};
\end{tikzpicture}
\caption{Generalization across occupancy states as provided by \citet{cunha2023convex}'s uniform continuity properties. Plot \textbf{A} describes generalization across all conditional occupancy states where the value function is given by a collection $\{\textcolor{sthlmBlue}{\omega}, \textcolor{sthlmRed}{\omega}, \textcolor{sthlmGreen}{\omega}\}$ of piecewise-linear and concave functions of conditional occupancy states, \cf Theorem \ref{thm:cunha}. Plot \textbf{B} describes generalization across any occupancy state $\textcolor{sthlmBlue}{\pmb{s}}$ given as a distribution over conditional occupancy states, such that value $\textcolor{sthlmBlue}{\omega(\textcolor{sthlmBlue}{\pmb{s}})}$ given by $ \textcolor{sthlmOrange}{\pmb{s^{\mathtt{m},\circ}(o^{\circ})}} \cdot \textcolor{sthlmBlue}{\omega(\textcolor{sthlmOrange}{\pmb{s^{\mathtt{c},o^{\circ}}}})} + \textcolor{sthlmPurple}{\pmb{s^{\mathtt{m},\circ}(o^{\circ})}} \cdot\textcolor{sthlmBlue}{\omega(\textcolor{sthlmPurple}{\pmb{s^{\mathtt{c},o^{\circ}}}})} + \textcolor{sthlmYellow}{\pmb{s^{\mathtt{m},\circ}(o^{\circ})}} \cdot\textcolor{sthlmBlue}{\omega(\textcolor{sthlmYellow}{\pmb{s^{\mathtt{c},o^{\circ}}}})}$ is also a convex combinations of values from conditional occupancy states, \cf Theorem \ref{thm:cunha:convex}. In this form, the piece-wise linear and concave functions of conditional occupancy states become linear functions. \textbf{Best viewed in color}.}
\label{figure:convex:function:representations}
\end{figure}  

Figure \ref{figure:convex:function:representations} (\textbf{A}) showcases piece-wise linear and concave functions across conditional occupancy states, as described by Theorem \ref{thm:cunha}. It should be noted that a change of basis effectively reveals the convexity property of the optimal value function across occupancy states. It will interest the reader to know that each piece-wise linear and concave function lower-bounds the optimal value function, \ie $\upsilon^*_\tau(s_\tau) \geq \sum_{o^{\pw}\in O^{\pw}}s^{\mathtt{m},\circ}_\tau(o^{\pw}) \min_{\alpha_\tau\in V_\tau} \alpha_\tau(s^{\mathtt{c},o^{\pw}}_\tau)$ for any collection of  linear functions $V_\tau\in \mathbb{V}^*_\tau$. So, as we populate a family $\mathbb{V}_\tau \subseteq \mathbb{V}^*_\tau$, the induced values lower-bound the optimal values. 

\begin{theorem}[Adapted from \citet{cunha2023convex}]
\label{thm:cunha:convex}
For any arbitrary $M'$, the optimal value functions $(\upsilon^*_0,\ldots, \upsilon^*_\ell)$ solutions of (\ref{eqn:bellman:optimality}) are maximum of linear functions of occupancy states, when occupancy states are expressed as distributions over conditional occupancy states, \ie there exists collections  $(\mathbb{W}^*_0,\ldots,\mathbb{W}^*_\ell)$ of  linear functions over occupancy states such that: for any stage $\tau$ 
\begin{align} 
\upsilon^*_\tau\colon s_\tau &\textstyle\mapsto \max_{\omega_\tau\in \mathbb{W}^*_\tau}  \omega_\tau(s_\tau).
\label{eqn:thm:greedy:convex}
\end{align}
\end{theorem}

\citet{cunha2023convex} presents a detailed proof of Theorem \ref{thm:cunha}. Although the property of uniform continuity it establishes is stronger than its predecessors, little is known about the maintenance of value function representations and the eventual discovery of optimal value functions. \citet{cunha2023convex} also demonstrated that when occupancy states are expressed as distributions over conditional occupancy states, the optimal value function is convex over occupancy states, \cf Theorem \ref{thm:cunha:convex}.  In this paper, we investigate the following question.

\tikzstyle{mybox} = [draw=black, very thick, rectangle, rounded corners, inner ysep=5pt, inner xsep=5pt]
\vspace{2pt}
\begin{tikzpicture}
\node [mybox] (box){
\begin{minipage}{.96\linewidth}
\quad\emph{How can we define efficient point-based operators $\mathbb{G}^{\pb}$ and $\mathbb{H}^{\pb}$ (resp. $\mathbb{G}^{\pw}$ and $\mathbb{H}^{\pw}$) of value functions represented using Theorem \ref{thm:cunha}, while ensuring the identification of an $\epsilon$-optimal joint policy for occupancy Markov game $M'$ (resp. $M$)?}
\end{minipage}
};
\end{tikzpicture}

\section{Exploiting Uniform Continuity}
This section uses the value function representation introduced in Theorem \ref{thm:cunha} to define the operators $\mathbb{G}^{\pb}$ and $\mathbb{H}^{\pb}$ (resp. $\mathbb{G}^{\pw}$ and $\mathbb{H}^{\pw}$) used to optimally solve $M'$. The main result of this section establishes the formal proof that $\mathbb{G}^{\pb} \colon (\upsilon^{\pb}_{\tau+1}, s_\tau)\mapsto a^{\pb}_\tau$ are solutions of linear programs with a polynomial number of constraints. We shall draw our attention to point-based operators, \ie operators build upon a sample set of occupancy states. The resulting value functions are built only upon this sample set of occupancy states, yet they generalize over the entire occupancy space thanks to the uniform continuity property stated in Theorem \ref{thm:cunha} from \citep{cunha2023convex}.

\begin{restatable}[]{thm}{thmgreedy}[Proof in Appendix \ref{proofthmgreedy}]
\label{thm:greedy}
Let $s_\tau$ be an occupancy state and  $\upsilon^{\pb}_{\tau+1}$ be the value function of the player ${\pb}$ at stage $\tau+1$.
The greedy decision rule $\mathbb{G}^{\pb}(s_\tau,\upsilon^{\pb}_{\tau+1})$ of the player ${\pb}$, \ie $\argmax_{a^{\pb}_\tau} \min_{a^{\pw}_\tau} R(s_\tau,a_\tau) + \gamma\upsilon^{\pb}_{\tau+1}(T(s_\tau,a_\tau))$ is the solution of the linear program in Figure \ref{fig:greedy:lp}. 
\end{restatable}

\begin{figure}[H]
\vspace{-2pt}
   \scriptsize
\tikzstyle{mybox} = [draw=black, very thick, rectangle, rounded corners, inner ysep=5pt, inner xsep=5pt]
\vspace{-4pt}
\begin{tikzpicture}
\node [mybox] (box){
\begin{minipage}{.95\linewidth}
\vspace{-12pt}
\begin{align*}
&\textbf{Maximize}\quad \textstyle \sum_{o^{\pw}\in O^{\pw}}  \textcolor{sthlmRed}{\alpha_\theta(o^{\pw})} \\
&\textbf{Subject to }\quad \\
&\quad\textstyle  \textcolor{sthlmRed}{\alpha_\theta(o^{\pw})} \leq   \sum_{V}\sum_{z^{\pw}} \textcolor{sthlmRed}{ \beta_{V}(o^{\pw}u^{\pw}z^{\pw})},~  \forall  \textcolor{sthlmGreen}{o^{\pw}\in O^{\pw}, u^{\pw}\in U^{\pw}}\\
&\quad\textstyle \textcolor{sthlmRed}{\beta_{V}(o^{\pw}u^{\pw}z^{\pw})} \leq    \sum_{o^{\pb}\in O^{\pb}} \sum_{u^{\pb}\in U^{\pb}}   g_{V}^{\alpha}(o,u,z^{\pw}) \cdot \textcolor{sthlmRed}{\theta(V,u^{\pb}|o^{\pb})},~ \forall \textcolor{sthlmGreen}{V \in \mathbb{V}, \alpha \in V, o^{\pw} \in O^{\pw}, u^{\pw} \in U^{\pw}, z^{\pw} \in Z^{\pw}} \\
&\quad\textstyle \sum_{V} \sum_{u^{\pb}} \textcolor{sthlmRed}{\theta(V,u^{\pb}|o^{\pb})} = 1,~ \forall \textcolor{sthlmGreen}{o^{\pb}\in O^{\pb}}\\
&\textbf{Variables }\quad \\
 &\quad \textcolor{sthlmRed}{\alpha_\theta(o^{\pw})} \in \mathbb{R},~ \forall o^{\pw}\in O^{\pw}\\
 &\quad \textcolor{sthlmRed}{\theta(V,u^{\pb}|o^{\pb})}\in [0,1],~ \forall V\in \mathbb{V}, u^{\pb}\in U^{\pb}, o^{\pb}\in O^{\pb}\\
 &\quad \textcolor{sthlmRed}{ \beta_{V}(o^{\pw}u^{\pw}z^{\pw})}\in \mathbb{R}, ~ \forall  V\in \mathbb{V}, o^{\pw}\in O^{\pw}, u^{\pw}\in U^{\pw}, z^{\pw}\in Z^{\pw}.
\end{align*}
\end{minipage}
};
\end{tikzpicture}
\vspace{-.15cm}
\caption{The linear program for the selection of greedy decision rule $\mathbb{G}^{\pb}(s_\tau,\upsilon^{\pb}_{\tau+1})$,  with function    $g_{V}^{\alpha} (o,u,z^{\pw}) =   \sum_x  s_\tau(x,o)  \sum_{y,z^{\pb}}p_{xy}^{uz}\cdot ( \frac{r_{xu}}{|Z^{\pw}|} + \gamma\alpha(y,o^{\pb}u^{\pb}z^{\pb})) $.  
The \textcolor{sthlmRed}{\bf red} quantities are variables; \textcolor{sthlmGreen}{\bf green} ones are constraint identifiers, and black ones are constants. \textbf{Best viewed in color}.}
\label{fig:greedy:lp}
\end{figure}

{
\begin{figure}[H]
\vspace{-2pt}
    \centering
   \scriptsize
\tikzstyle{mybox} = [draw=black, very thick, rectangle, rounded corners, inner ysep=5pt, inner xsep=5pt]
\vspace{-4pt}
\begin{tikzpicture}
\node [mybox] (box){
\begin{minipage}{.96\linewidth}
\vspace{-2pt}
\begin{align}
&\textstyle V_{\textcolor{sthlmRed}{\theta}} \doteq \{w^{\bar{s}_\tau,o^{\pw}} | \bar{s}_\tau\in \bar{S}_\tau, o^{\pw} \in \bar{O}^{\pw}_\tau(\bar{s}_\tau)\}
\\
&\textstyle w^{s_\tau,o^{\pw}} (x, o^\pb) = \argmin_{w^{s_\tau,o^{\pw},u^{\pw}}} \sum_{x,o} s_\tau(x,o) \cdot w_{V,\alpha}^{s_\tau,o^{\pw}u^{\pw}z^{\pw}}(x,o^{\pb}),  \forall \textcolor{sthlmGreen}{o^{\pw}\in \bar{O}^{\pw}_\tau}
\\
&\textstyle w^{s_\tau,o^{\pw}, u^\pw} (x,o^{\pb}) = \sum_{V \in \mathbb{V}} \sum_{z^\pw} w_V^{s_\tau,o^{\pw}u^{\pw}z^{\pw}} (x,o^\pb),  \forall \textcolor{sthlmGreen}{o^{\pw}\in \bar{O}^{\pw}_\tau, u^{\pw}\in U^{\pw}}
\\
&\textstyle w_V^{s_\tau,o^{\pw}u^{\pw}z^{\pw}} (x,o^{\pb}) = \argmin_{w_{V,\alpha}^{o^{\pw}u^{\pw}z^{\pw}}} \sum_{x,o}s_\tau(x,o) \cdot w_{V,\alpha}^{o^{\pw}u^{\pw}z^{\pw}}(x,o^{\pb}),  \forall \textcolor{sthlmGreen}{V\in \mathbb{V}, o^{\pw}\in \bar{O}^{\pw}_\tau, u^{\pw}\in U^{\pw}, z^{\pw}\in Z^{\pw}}
\\
&w_{V,\alpha}^{o^{\pw}u^{\pw}z^{\pw}} (x,o^{\pb}) =\!\!\! \sum_{u^{\pb},z^{\pb}}\sum_y \textcolor{sthlmRed}{\theta(V,u^{\pb}|o^{\pb})}   p_{xy}^{uz}(\frac{r_{xu}}{|Z^{\pw}|} + \gamma\alpha(y,o^{\pb}u^{\pb}z^{\pb})), \forall \textcolor{sthlmGreen}{V \in \mathbb{V},\alpha\in V, o^{\pw}\in \bar{O}^{\pw}_\tau, u^{\pw}\in U^{\pw}, z^{\pw}\in Z^{\pw}}
\end{align}
\end{minipage}
};
\end{tikzpicture}
\vspace{-6pt}
\caption{Computing $V_{\textcolor{sthlmRed}{\theta}}$ from the solution $\textcolor{sthlmRed}{\theta}$ of the linear program of Figure \ref{fig:greedy:lp}. \textbf{Best viewed in color}.}
\label{fig:greedy:lp:retrival}
\end{figure}
}

Theorem \ref{thm:greedy} specifies linear programs necessary to implement the point-based greedy-action selection operator $\mathbb{G}^{\pb}$. It is worth noticing that the number of constraints in these linear programs is polynomial in the size of the value function $\upsilon^{\pb}_{\tau+1}\colon s_{\tau+1}\mapsto \max_{V\in \mathbb{V}_{\tau+1}} \sum_{o^{\pw}} s_{\tau+1}^{\mathtt{m},\circ}(o^{\pw}) \min_{\alpha\in V} \alpha(s_{\tau+1}^{\mathtt{c},o^{\pw}})$. To better understand this, if we let $\bar{O}_\tau^{\pw}(s_\tau) = \{o^{\pw} | o^{\pw}\in O^{\pw}_\tau,s_\tau^{\mathtt{m},\circ}(o^{\pw})>0 \}$ and $\bar{V} \in \argmax_{V\in \mathbb{V}_{\tau+1}} |V|$ then the linear program in Figure \ref{fig:greedy:lp} involves about $\pmb{O}(|\mathbb{V}_{\tau+1}||\bar{V}||\bar{O}_\tau^{\pw}(s_\tau)||U^{\pw}||Z^{\pw}|)$ constraints.  
After defining how to select the greedy decision rules for either player, resulting in a joint greedy decision rule, we now present an implementation of the \citeauthor{bellman}'s update operator for player $\pb$ $\mathbb{H}^{\pb}\colon (\upsilon^{\pb}_{\tau+1}, s_\tau, \mathbb{G}^{\pb} (\upsilon^{\pb}_{\tau+1}, s_\tau)) \mapsto \upsilon^{\pb}_\tau$. A point-based update operator of a family of collections of linear functions is sound if it improves the current value function at least in one occupancy state but it overestimates the optimal value function nowhere.

\begin{restatable}[]{thm}{thmsimstomgbellmanoptimalityeqn}[Proof in Appendix \ref{proofthmsimstomgbellmanoptimalityeqn}]
\label{thmsimstomgbellmanoptimalityeqn}
Let  $s_\tau$ be an occupancy state, $\upsilon^{\pb}_\tau$ and $\upsilon^{\pb}_{\tau+1}$ be value functions of the player ${\pb}$ given by $\mathbb{V}_\tau$ and $\mathbb{V}_{\tau+1}$, respectively. If we let $\theta  = \mathbb{G}^{\pb} (\upsilon^{\pb}_{\tau+1}, s_\tau)$ be the solution of the greedy operator at occupancy state $s_\tau$ for player ${\pb}$, and $V_{\theta}$ be the collection of linear functions computed as in Figure \ref{fig:greedy:lp:retrival}, then  the updated value function $\mathbb{H}^{\pb} (\upsilon^{\pb}_{\tau+1}, s_\tau, \theta)$ at stage $\tau$, given by $\mathbb{V}_\tau = \mathbb{V}_\tau \cup \{V_{\theta}\}$, is sound.
\end{restatable}

Notice that as the updates proceed families of collections $\mathbb{V}_{0:} = (\mathbb{V}_0, \mathbb{V}_1, \ldots, \mathbb{V}_{\ell-1})$ may become cumbersome. For this reason, we investigate pruning techniques tailored for families of collections.
A collection of linear functions $V_\tau$ is said to be dominated by collections in the family of collections $\mathbb{V}_\tau$ if its removal does not affect the value function $\upsilon_\tau\colon s_\tau \mapsto \max_{V_\tau\in \mathbb{V}_\tau} \sum_{o^{\pw}\in O^{\pw}} s_\tau^{\mathtt{m},\pw}(o^{\pw}) \min_{\alpha_\tau\in V_\tau} \alpha_\tau(s_\tau^{\mathtt{c},o^{\pw}})$. We rely on bounded pruning, \cf Algorithm \ref{pruning:zsposg} in Appendix, to remove point-based dominated collections in a given family and a sample set of occupancy states. To somewhat mitigate the growth of the families of collections, we also prune occupancy states that support no collections.

\section{Point-Based Value Iteration}

This section discusses how the point-based value-iteration (PBVI) algorithm \citep{pineau2003point} can be adapted to compute an $\epsilon$-optimal joint policy for $M'$ (or $M$) starting from an initial state distribution $s_0$, for a planning horizon $\ell$. The PBVI algorithm was chosen because it is guaranteed to find near-optimal solutions asymptotically. It is worth noting that our update operators are not designed to explore the occupancy-state space optimistically. This is because our uniform continuity property is intended to represent only pessimistic value functions. Therefore, algorithms such as heuristic search value iteration \citep{SmithS06}, which account for optimistic explorations, cannot be used. 

\begin{wrapfigure}{r}{0.41\textwidth} 
    \begin{minipage}{\linewidth}
        \begin{algorithm}[H]
        \caption{PBVI for $M'$.}
        \label{pbvi:zsposg}
        \begin{algorithmic}
            \STATE ${\mathtt{function}~ \mathtt{PBVI}()}$
            \STATE Initialize $\bar{S}_{0:}$, $\mathbb{V}_{0:}$ and $\upsilon^{\pb}_{0:}$.
            \WHILE{has not converged}
            \STATE $\mathtt{improve}(\mathbb{V}_{0:})$.
            \STATE $\bar{S}_{0:} \gets \mathtt{expand}(\bar{S}_{0:})$.
            \ENDWHILE
            \STATE ${\mathtt{function}~ \mathtt{improve}()}$
            \FOR{$\tau=\ell-1$ to $0$}
            \FOR{$s_\tau \in \bar{S}_\tau$}
            \STATE $\mathbb{V}_\tau \gets \mathbb{V}_\tau  \cup \{V_{\mathbb{G}^{\pb} (\upsilon^{\pb}_{\tau+1}, s_\tau)} \}$.
            \ENDFOR
            \ENDFOR
        \end{algorithmic}
        \end{algorithm}
    \end{minipage}
\end{wrapfigure}

The PBVI algorithm, \cf Algorithm \ref{pbvi:zsposg}, has two main parts for solving $M'$ (or $M$). First, it uses a finite and reachable occupancy-state subset $\bar{S}_\tau$ to bound the size of the value functions for either player at each stage $\tau$ of the game. Next, it optimizes the value functions represented in a collection $\mathbb{V}_\tau$ at each stage $\tau$ using point-based backups. This means backups are executed in no particular order, and each iteration traverses occupancy-state subsets bottom-up. The process repeats until convergence when the difference between $\upsilon^{\pb}_0(s_0)$ and $\upsilon^{\pw}_0(s_0)$ is less than or equal to a certain threshold ($\epsilon$), or until a budget (such as CPU time, memory, or the number of iterations) has been exhausted. The algorithm adds supplemental points into occupancy subsets to improve the value functions further. It selects candidate points using a portfolio of exploration strategies such as random explorations and greedy approaches with respect to the underlying (PO)MDP value functions. For each stage $\tau$, the algorithm only adds candidate points beyond a certain distance from the occupancy subset $\bar{S}_\tau$ to create a new occupancy-state set $\bar{S}_{\tau+1}$. 
For any arbitrary occupancy-state subsets $\bar{S}_{0:}$,  PBVI  produces value $\upsilon^{\pb}_0(s_0)$.  The error between $\upsilon^{\pb}_0(s_0)$ and $\upsilon^*_{0}(s_0)$ is bounded. The bound depends on how $\bar{S}_{0:}$ samples the entire occupancy-state space; with denser sampling, the estimate $\upsilon^{\pb}_0(s_0)$ converges to $\upsilon^*_0(s_0)$. The remainder of this section states and proves our error bound.

Define the density $\delta_{ \bar{S}_{0:}}$ to be the maximum distance from any legal occupancy state to subsets $\bar{S}_{0:}$. More precisely, $\delta_{ \bar{S}_{0:}} \doteq \max_{\tau\in \llbracket 0:\ell-1\rrbracket}\max_{s\in S_\tau}\min_{s'\in \bar{S}_\tau} \|s-s'\|_1$. Define a positive scalar $c$ such that $\|r(\cdot,\cdot)\|_\infty \leq c$.

\begin{restatable}[]{thm}{thmerrorbound}[Proof in Appendix \ref{proofthmerrorbound}]
\label{thmerrorbound}
For any occupancy subsets $\bar{S}_{0:}$, the error of the PBVI algorithm is bounded by $\upsilon^*_0(s_0) - \upsilon^{\textcolor{gray}{\bullet}}_0(s_0) \leq  2c \delta_{ \bar{S}_{0:}} \frac{1+\ell \gamma^{\ell+1}-(\ell+1) \gamma^\ell }{(1-\gamma)^2}$.
\end{restatable}

Notice that the error that the PBVI algorithm does is $2c \delta_{ \bar{S}_{0:}}/(1-\gamma)^2$ whenever $\ell\to\infty$, which is equivalent to the error bound from \citet{pineau2003point}.

\section{Experiments}
We test our proposed method's performance on several well-known benchmarks: Adversarial Tiger, Competitive Tiger, Recycling, Mabc, and Matching Pennies. These are among challenging benchmarks available in the literature of partially observable stochatic games, for which we refer the reader to \url{http://masplan.org/} for an in-depth description. Many of them are well-known common-payoff benchmark problems adapted to our competitive setting by making player $\circ$ minimize (rather than maximize) the objective function. Further experimental details are provided in Appendix \ref{sec:computational_resources}.

\begin{wraptable}{r}{0.6\textwidth}
    \caption{Snapshot of empirical results. For each game and algorithm, we report time (in hours) and the best value for horizons $\ell \in \{2,3,4,5,7,10\}$. {\sc oot} means a time limit of 2 hours has been exceeded, while {\sc oom} means that the algorithms ran out of memory}
    \label{table:main_results}
    \resizebox{\linewidth}{!}{
        \begin{tabular}{@{}c c rr rr rr rr rr rr} 
        
        \toprule%
         \centering%
           Problem
         & $\ell$
         & \multicolumn{2}{c}{{{\bfseries PBVI$_3$}}}
         & \multicolumn{2}{c}{{{\bfseries PBVI$_2$}}}
         & \multicolumn{2}{c}{{{\bfseries PBVI$_1$}}}
         & \multicolumn{2}{c}{{{\bfseries HSVI \citep{DelBufDibSaf-DGAA-23}}}}
         & \multicolumn{2}{c}{{{\bfseries CFR+ \citep{Tammelin14}}}}
         \\

        \cmidrule[0.4pt](r{0.125em}){1-2}%
        \cmidrule[0.4pt](lr{0.125em}){3-4}%
        \cmidrule[0.4pt](lr{0.125em}){5-6}%
        \cmidrule[0.4pt](lr{0.125em}){7-8}%
        \cmidrule[0.4pt](lr{0.125em}){9-10}%
        \cmidrule[0.4pt](lr{0.125em}){11-12}%
        \cmidrule[0.4pt](lr{0.125em}){13-14}%

            & 2 & 0 & \textcolor{sthlmRed}{\bf-0.4} & 0 & \textcolor{sthlmRed}{\bf-0.4} & 0 & \textcolor{sthlmRed}{\bf-0.4} & 0 & \textcolor{sthlmRed}{\bf-0.4} & 0 & \textcolor{sthlmRed}{\bf-0.4} \\
            & 3 & 0 & \textcolor{sthlmRed}{\bf-0.56} & 0 & \textcolor{sthlmRed}{\bf-0.56} & 0 & \textcolor{sthlmRed}{\bf-0.56} & 0 & \textcolor{sthlmRed}{\bf-0.56} & 0 & \textcolor{sthlmRed}{\bf-0.56} \\
            & 4 & 0 & -0.76 & 0.02 & -0.76 & 0.02 & -0.71 & OOT & [-1.89, 0.11] & 0 & \textcolor{sthlmRed}{\bf-0.75} \\
            & 5 & 0.07 & \textcolor{sthlmRed}{\bf-0.96} & 0.07 & \textcolor{sthlmRed}{\bf-0.96} & 1.7 & \textcolor{sthlmRed}{\bf-0.96} & OOT & [-3.54, 1.13] & -- & OOM \\
            & 7 & 1.3 & -1.36 & 0.1 & \textcolor{sthlmRed}{\bf-1.4} & 2 & -1.36 & OOT & [-4.4, 2.24] & -- & OOM \\
            \multirow{-6}{*}{\shortstack[C]{Adversarial \\ Tiger }} & 10 & 0.18 & -1.99 & 0.16 & \textcolor{sthlmRed}{\bf-2} & 2.8 & -1.98 & 1 & [-42, 24.6] & -- & OOM \\
        
        \hline
            & 2 & 0 & -0.01 & 0 & -0.015 & 0 & -0.015 & 0 & 0. & 0 & \textcolor{sthlmRed}{\bf0} \\
            & 3 & 0 & -0.06 & 0 & -0.033 & 0 & -0.033 & 1 & [-0.45,0.45] & 0.11 & \textcolor{sthlmRed}{\bf0} \\
            & 4 & 0.11 & -0.06 & 0.06 & \textcolor{sthlmRed}{\bf-0.056} & 0.03 & \textcolor{sthlmRed}{\bf-0.056} & OOT & [-1.11, 1.11] & -- & OOM \\
            & 5 & 1.5 & -0.07 & 0.1 & \textcolor{sthlmRed}{\bf-0.055} & 0.05 & -0.079 & OOT & [-1.84, 1.84] & -- & OOM \\
            & 7 & OOT & -0.14 & OOT & -0.16 & 3 & \textcolor{sthlmRed}{\bf-0.15} & OOT & [-7.31, 7.31] & -- & OOM \\
            \multirow{-6}{*}{\shortstack[C]{Competitive \\ Tiger }} & 10 & OOT & -0.2 & OOT & -0.18 & OOT & \textcolor{sthlmRed}{\bf-0.22} & 1 & [-52, 52] & -- & OOM \\
            
        \hline
            & 2 & 0 & \textcolor{sthlmRed}{\bf0.26} & 0 & \textcolor{sthlmRed}{\bf0.26} & 0 & \textcolor{sthlmRed}{\bf0.26} & 0 & \textcolor{sthlmRed}{\bf0.26} & 0 & \textcolor{sthlmRed}{\bf0.26} \\
            & 3 & 0 & \textcolor{sthlmRed}{\bf0.32} & 0 & \textcolor{sthlmRed}{\bf0.32} & 0.03 & 0.32 & 0.03 & 0.32 & 0 & \textcolor{sthlmRed}{\bf0.32} \\
            & 4 & 0 & \textcolor{sthlmRed}{\bf0.36} & 0 & \textcolor{sthlmRed}{\bf0.36} & 0.4 & \textcolor{sthlmRed}{\bf0.36} & OOT & [-0.04, 0.78] & 0.3 & \textcolor{sthlmRed}{\bf0.36} \\
            & 5 & 0 & \textcolor{sthlmRed}{\bf0.4} & 0.08 & \textcolor{sthlmRed}{\bf0.4} & 0.1 & \textcolor{sthlmRed}{\bf0.4} & OOT & [-0.49, 1.39] & -- & OOM \\
            & 7 & 0.12 & \textcolor{sthlmRed}{\bf0.48} & 3 & \textcolor{sthlmRed}{\bf0.48} & 0.31 & 0.47 & OOT & [-1.3, 2.5] & -- & OOM \\
            \multirow{-6}{*}{\shortstack[C]{Recycling }} & 10 & OOT & 0.58 & OOT & 0.59 & 0.7 & \textcolor{sthlmRed}{\bf0.6} & OOT & [ 16.2, -10.7] & -- & OOM \\
        
        \hline
            & 2 & 0 & \textcolor{sthlmRed}{\bf0.078} & 0 & \textcolor{sthlmRed}{\bf0.078} & 0 & \textcolor{sthlmRed}{\bf0.078} & 0 & \textcolor{sthlmRed}{\bf0.78} & 0 & \textcolor{sthlmRed}{\bf0.08} \\
            & 3 & 0 & \textcolor{sthlmRed}{\bf0.097} & 0 & \textcolor{sthlmRed}{\bf0.098} & 0 & \textcolor{sthlmRed}{\bf0.097} & 0 & \textcolor{sthlmRed}{\bf0.098} & 0 & \textcolor{sthlmRed}{\bf0.098} \\
            & 4 & 0.03 & \textcolor{sthlmRed}{\bf0.11} & 0.03 & \textcolor{sthlmRed}{\bf0.11} & 0.03 & \textcolor{sthlmRed}{\bf0.11} & OOT & [0.1, 0.15] & 0 & \textcolor{sthlmRed}{\bf0.11} \\
            & 5 & 0.08 & \textcolor{sthlmRed}{\bf0.12} & 0 & \textcolor{sthlmRed}{\bf0.12} & 0.1 & \textcolor{sthlmRed}{\bf0.12} & 1.6 & [0.08, 0.27] & 0.1 & \textcolor{sthlmRed}{\bf0.12} \\
            & 7 & 0 & \textcolor{sthlmRed}{\bf0.14} & 0 & \textcolor{sthlmRed}{\bf0.14} & 2.8 & \textcolor{sthlmRed}{\bf0.14} & OOT & [0.08, 0.5] & -- & OOM \\
            \multirow{-6}{*}{\shortstack[C]{Mabc}} & 10 & 0.1 & \textcolor{sthlmRed}{\bf0.16} & 0.4 & 0.17 & 0.3 & 0.17 & OOT & [3.2, 0.35] & -- & OOM \\
        
        \hline
            & 2 & 0 & $\textcolor{sthlmRed}{\bf0.2}$ & 0 & $\textcolor{sthlmRed}{\bf0.2}$ & 0 & $\textcolor{sthlmRed}{\bf0.2}$ & 0 & $\textcolor{sthlmRed}{\bf0.2}$ & 0 & $\textcolor{sthlmRed}{\bf0.2}$ \\
            & 3 & 0 & 0.37 & 0 & $\textcolor{sthlmRed}{\bf0.4}$ & 0 & $\textcolor{sthlmRed}{\bf0.4}$ & 0 & $\textcolor{sthlmRed}{\bf0.4}$ & 0 & $\textcolor{sthlmRed}{\bf0.4}$ \\
            & 4 & OOT & 0.61 & 0 & $\textcolor{sthlmRed}{\bf0.6}$ & 0 & $\textcolor{sthlmRed}{\bf0.6}$ & 0 & $\textcolor{sthlmRed}{\bf0.6}$ & 0 & $\textcolor{sthlmRed}{\bf0.6}$ \\
            & 5 & OOT & 1.17 & 0.1 & $\textcolor{sthlmRed}{\bf0.8}$ & 0 & $\textcolor{sthlmRed}{\bf0.8}$ & OOT & $\textcolor{sthlmRed}{\bf0.8}$ & 0 & $\textcolor{sthlmRed}{\bf0.8}$ \\
            & 7 & OOT & 1.19 & 0 & 1.157 & 0 & $\textcolor{sthlmRed}{\bf1.2}$ & OOT & [-1.7, 3.95] & 0 & $\textcolor{sthlmRed}{\bf1.2}$ \\
            \multirow{-6}{*}{\shortstack[C]{Matching \\ Pennies }} & 10 & OOT & \textcolor{sthlmRed}{\bf1.8} & 0 & \textcolor{sthlmRed}{\bf1.8} & 0 & \textcolor{sthlmRed}{\bf1.8} & -- & -- & 0.4 & \textcolor{sthlmRed}{\bf1.8} \\
        \bottomrule
        \end{tabular}
        \vspace{-.5cm}
    }
\end{wraptable}

For each problem, we report the performance of three versions of the PBVI algorithm: PBVI$_1$, which is the most basic version of the algorithm as it does not rely on any pruning, and PBVI$_2$ and PBVI$_3$ which both implement the bounded pruning algorithm presented in Algorithm \ref{pruning:zsposg}. Pruning is either applied on the collection of linear functions $V_\tau$, or on the points that are linked to such functions. Their performance is compared against that obtained by \citet{DelBufDibSaf-DGAA-23}, where the HSVI algorithm was used, and that obtained by the CFR+ algorithm of \citet{Tammelin14}. 
Table \ref{table:main_results} shows a panoramic of our results: we report the best value obtained by each algorithm (in pink) on the aforementioned benchmarks for several horizons $\ell$ of increasing complexity. Specifically, all algorithms were tested on $\ell \in \{2,3,4,5,7,10 \}$. Note that the larger the value of $\ell$, the more memory- or time-demanding the task becomes, which might prevent the algorithms from converging. If this is the case, we report the final performance of the algorithms as OOM and OOT, respectively. We can see from the reported results that, overall, all versions of the PBVI algorithm are able to tackle the aforementioned benchmarks successfully. Specifically, we can note that the PBVI$_2$ variation of the algorithm performs on par with its PBVI$_1$ counterpart across most games and horizons $\ell$, therefore highlighting the benefits of including pruning to the algorithm (see again Appendix \ref{pruning:zsposg} for further details). However, there are also cases where PBVI$_3$ shows promising results (see, for example, the results obtained on Recycling and Box Pushing). Yet, it is important to note that these results are not as robust as the ones obtained by PBVI$_2$, as demonstrated by the performance obtained on the Matching Pennies game. 

Overall, one can identify the main strength of our proposed approach in its scalability, as, to the best of our knowledge, this is the first example where such benchmarks can be solved even for very large horizons, a scenario that the techniques proposed by \citet{DelBufDibSaf-DGAA-23} and \citet{Tammelin14} are not capable of dealing with. We visually represent the scalability of our PBVI variants across five games in Figure \ref{fig:main_results}, where the challenging horizon of $\ell=10$ is considered. The capability of our PBVI variants to rapidly find optimal solutions comparable to HSVI is also depicted in the plots presented in Appendix \ref{sec:appendix:additional:plots}, where the same games presented in Figure \ref{fig:main_results} are reported but for $\ell=4$. Additionally, we also successfully applied PBVI$_1$ to address specific subclasses of zs-POSGs, including fully-observable zs-POSGs (MGs) and zs-POSGs with public actions and observations (POMGs), as detailed in Appendix \ref{sec:appendix:subclasses}. All these results illustrate the adaptability of our methods to subclasses of zs-POSGs while maintaining high efficiency. In certain domains, the scalability of our PBVI variants to larger games and planning horizons is constrained by the considerable number of iterations (and consequently samples) required. Point pruning techniques can be employed, like in PBVI$_3$, although this comes at the cost of sacrificing theoretical guarantees.

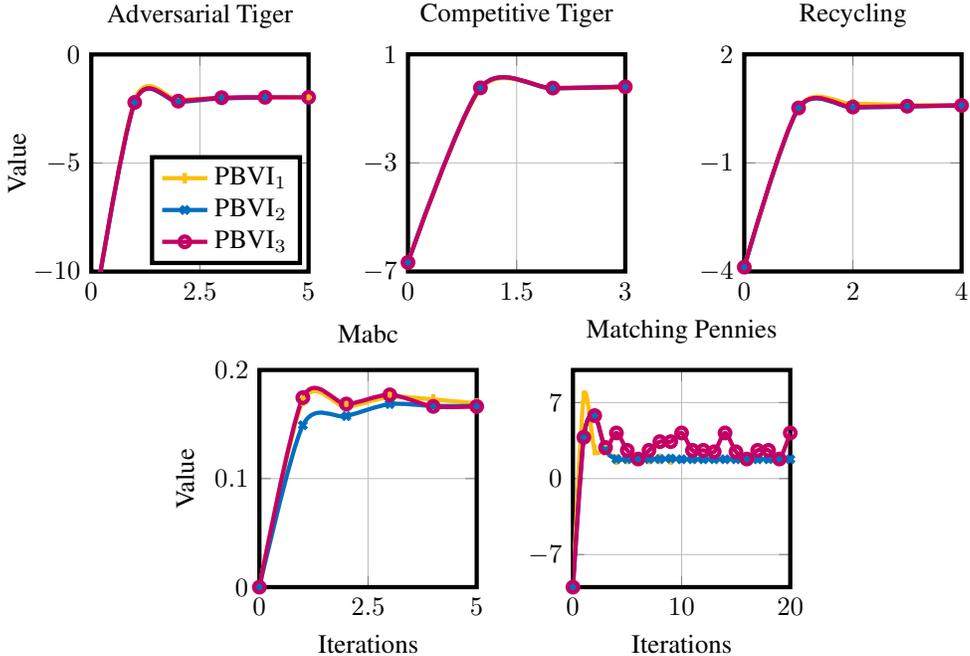
\begin{figure}[ht]
  \centering
  \begin{subfigure}{0.32\textwidth}
    \centering
    \begin{tikzpicture}
      \begin{axis}[
        title={Adversarial Tiger},
        width=\linewidth,
        height=\linewidth,
        ylabel={Value},
        ylabel style={at={(axis description cs:-0.25,.5)}},
        xmin=0, xmax=5,
        xtick={0, 2.5, 5}, 
        ymin=-10, ymax=0,
        ytick={-10, -5, 0},
        grid=both,
        style={ultra thick},
        legend pos=south east,
        ]
        \addplot+[smooth,mark=|,sthlmYellow] table [col sep=comma, x=iter, y=value] {./paper/results/PBVI_without_pruning/H10/adversarialtiger.csv};
        \addplot+[smooth,mark=x,sthlmBlue] table [col sep=comma, x=iter, y=value] {./paper/results/PBVI_pruning_hyperplane/H10/adversarialtiger.csv};
        \addplot+[smooth,mark=o,sthlmRed] table [col sep=comma, x=iter, y=value] {./paper/results/PBVI_pruning_hyperplane_point/H10/adversarialtiger.csv};
        \legend{PBVI$_1$, PBVI$_2$, PBVI$_3$}
      \end{axis}
    \end{tikzpicture}
    \label{fig:lineplots:pbvi:adversarial_tiger}
  \end{subfigure}%
  \begin{subfigure}{0.32\textwidth}
    \centering
    \begin{tikzpicture}
      \begin{axis}[
        title={Competitive Tiger},
        width=\linewidth,
        height=\linewidth,
        xmin=0, xmax=3,
        xtick={0, 1.5, 3}, 
        ymin=-7, ymax=1,
        ytick={-7, -3, 1},
        grid=major,
        style={ultra thick},
        ]
        \addplot+[smooth,mark=|,sthlmYellow] table [col sep=comma, x=iter, y=value] {./paper/results/PBVI_without_pruning/H10/competitivetiger.csv};
        \addplot+[smooth,mark=x,sthlmBlue] table [col sep=comma, x=iter, y=value] {./paper/results/PBVI_pruning_hyperplane/H10/competitivetiger.csv};
        \addplot+[smooth,mark=o,sthlmRed] table [col sep=comma, x=iter, y=value] {./paper/results/PBVI_pruning_hyperplane_point/H10/competitivetiger.csv};
      \end{axis}
    \end{tikzpicture}
    \label{fig:lineplots:pbvi:competitive_tiger}
  \end{subfigure}%
  \begin{subfigure}{0.32\textwidth}
    \centering
    \begin{tikzpicture}
      \begin{axis}[
        title={Recycling},
        width=\linewidth,
        height=\linewidth,
        ylabel style={yshift=-5mm},
        xmin=0, xmax=4,
        xtick={0, 2, 4}, 
        ymin=-4, ymax=2,
        ytick={-4, -1, 2},
        grid=both,
        style={ultra thick},
        ]
        \addplot+[smooth,mark=|,sthlmYellow] table [col sep=comma, x=iter, y=value] {./paper/results/PBVI_without_pruning/H10/recycling.csv};
        \addplot+[smooth,mark=x,sthlmBlue] table [col sep=comma, x=iter, y=value] {./paper/results/PBVI_pruning_hyperplane/H10/recycling.csv};
        \addplot+[smooth,mark=o,sthlmRed] table [col sep=comma, x=iter, y=value] {./paper/results/PBVI_pruning_hyperplane_point/H10/recycling.csv};
      \end{axis}
    \end{tikzpicture}
    \label{fig:lineplots:pbvi:recycling}
  \end{subfigure}
  \begin{subfigure}{0.32\textwidth}
    \centering
    \begin{tikzpicture}
      \begin{axis}[
        title={Mabc},
        width=\linewidth,
        height=\linewidth,
        xlabel={Iterations},
        ylabel={Value},
        ylabel style={at={(axis description cs:-0.25,.5)}},
        xmin=0, xmax=5,
        xtick={0, 2.5, 5}, 
        ymin=0, ymax=0.2,
        ytick={0, 0.1, 0.2},
        grid=both,
        style={ultra thick},
        ]
        \addplot+[smooth,mark=|,sthlmYellow] table [col sep=comma, x=iter, y=value] {./paper/results/PBVI_without_pruning/H10/mabc.csv};
        \addplot+[smooth,mark=x,sthlmBlue] table [col sep=comma, x=iter, y=value] {./paper/results/PBVI_pruning_hyperplane/H10/mabc.csv};
        \addplot+[smooth,mark=o,sthlmRed] table [col sep=comma, x=iter, y=value] {./paper/results/PBVI_pruning_hyperplane_point/H10/mabc.csv};

      \end{axis}
    \end{tikzpicture}
    \label{fig:lineplots:pbvi:mabc}
  \end{subfigure}%
  \begin{subfigure}{0.32\textwidth}
    \centering
    \begin{tikzpicture}
      \begin{axis}[
        title={Matching Pennies},
        width=\linewidth,
        height=\linewidth,
        xlabel={Iterations},
        xmin=0, xmax=20,
        xtick={0, 10, 20}, 
        ymin=-10, ymax=10,
        ytick={-7, 0, 7},
        grid=both,
        style={ultra thick},
        ]
        \addplot+[smooth,mark=|,sthlmYellow] table [col sep=comma, x=iter, y=value] {./paper/results/PBVI_without_pruning/H10/matchingpennies.csv};
        \addplot+[smooth,mark=x,sthlmBlue] table [col sep=comma, x=iter, y=value] {./paper/results/PBVI_pruning_hyperplane/H10/matchingpennies.csv};
        \addplot+[smooth,mark=o,sthlmRed] table [col sep=comma, x=iter, y=value] {./paper/results/PBVI_pruning_hyperplane_point/H10/matchingpennies.csv};
      \end{axis}
    \end{tikzpicture}
    \label{fig:lineplots:pbvi:matching_pennies}
  \end{subfigure}%
  \caption{A visual representation of the performance of our proposed algorithms for the challenging horizon $\ell = 10$ across five different games. We can see that for Adversarial Tiger, Competitive Tiger and Recycling all three PBVI variants perform equally. PBVI$_2$ also matches the performance of PBVI$_1$ on Mabc and Matching Pennies, with the latter benchmark being harder to solve for PBVI$_3$. \textbf{Best viewed in color.}}
  \label{fig:main_results}
\end{figure}

\section{Conclusion}

We have introduced an offline point-based value iteration algorithm for finding near-optimal zs-POSG solutions. It leverages a recently discovered uniform continuity property, allowing for the design of update and backup operators with significantly reduced complexity. This approach yields point-based value iteration algorithms capable of scaling to larger games and planning horizons while retaining theoretical guarantees. Finally, we also provided the first offline algorithm, building upon \citeauthor{bellman}'s principle of optimality, which can effectively scale to large games. Experimental results conducted on various domains support these claims.      

The utilization of dynamic programming presents an appealing approach by decomposing the original zs-POSG into smaller subgames, subsequently solving them optimally in reverse order. This methodology also enables the extrapolation of values from one subgame to another through uniform continuity properties. These two attributes mutually reinforce each other while upholding theoretical guarantees. To the best of our knowledge, this work represents the first algorithmic framework for zs-POSGs that concurrently demonstrates such characteristics.  

The utilization of the dynamic programming approach is significant due to its association with reinforcement learning. Recent studies have showcased the capability to address reinforcement learning operators related to common-payoff occupancy Markov games (common-payoff OMGs), as evidenced by \citet{pmlr-v80-dibangoye18a}. This paper is anticipated to serve as a fundamental basis for further advancements in dynamic programming and reinforcement learning methodologies for zs-POSGs as zero-sum occupancy Markov games (zero-sum OMGs).

\bibliographystyle{plainnat}
\bibliography{mybib}

\appendix

\section{Proof of Theorem \ref{thm:greedy}}
\label{proofthmgreedy}
\thmgreedy*
\begin{proof}
The proof starts with the definition of the greedy-action selection at occupancy state $s_\tau$, and value function $\upsilon^{\pb}_{\tau+1}\colon s_{\tau+1}\mapsto \max_{V\in \mathbb{V}} \sum_{o^{\pw} \in O^{\pw}} s_{\tau+1}^{\mathtt{m},\circ}(o^{\pw}) \min_{\alpha\in V} \alpha(s_{\tau+1}^{\mathtt{c},o^{\pw}})$. Let $q^{\pb}_\tau \colon (s_\tau,a_\tau) \mapsto R(s_\tau,a_\tau) + \gamma\upsilon^{\pb}_{\tau+1}(T(s_\tau,a_\tau))$. Then, 
$\mathbb{G}^{\pb}(s_\tau,\upsilon^{\pb}_{\tau+1}) =\textstyle \argmax_{a_\tau^{\pb}}\min_{a_\tau^{\pw}} ~ q^{\pb}_\tau(s_\tau,a_\tau)$.
Let $s_{\tau+1}^{\mathtt{m},\pw}(o_{\tau+1}^{\pw})$ and $s_{\tau+1}^{\mathtt{c}, o_{\tau+1}^{\pw}}$ be marginal probability and conditional occupancy state associated with $T(s_\tau,a_\tau)$.  The injection of  (\ref{eqn:thm:greedy}) into $\upsilon^{\pb}_{\tau+1}(T(s_\tau,a_\tau))$ produces the following expression:
\begin{align*} 
\upsilon^{\pb}_{\tau+1}(T(s_\tau,a_\tau))&=\max_{V\in \mathbb{V}} {\textstyle \sum_{o_{\tau+1}^{\pw} \in O_{\tau+1}^{\pw}}}s_{\tau+1}^{\mathtt{m},\pw}(o_{\tau+1}^{\pw}) \min_{\alpha\in V}\alpha(s_{\tau+1}^{\mathtt{c},o_{\tau+1}^{\pw}}).
\end{align*}
If we replace $\max_{V\in \mathbb{V}}$ by $\max_{\xi \in \triangle(\mathbb{V})}$  then there is no loss in optimality, \ie
\begin{align*} 
&=\max_{\xi \in \triangle(\mathbb{V})} {\textstyle \sum_{o_{\tau+1}^{\pw} \in O_{\tau+1}^{\pw}}  \sum_{V \in \mathbb{V}}}\min_{\alpha \in V}\xi(V) s_{\tau+1}^{\mathtt{m},\pw}(o_{\tau+1}^{\pw}) \alpha(s_{\tau+1}^{\mathtt{c},o_{\tau+1}^{\pw}}).
\end{align*}
\citet{cunha2023convex} shows that these statistics $s_{\tau+1}^{\mathtt{m},\pw}(o_{\tau+1}^{\pw})$ and $s_{\tau+1}^{\mathtt{c}, o_{\tau+1}^{\pw}}$  depend on $T(s_\tau,a_\tau)$ only through marginal probability $s_\tau^{\mathtt{m},\pw}(o_\tau^{\pw})$, conditional occupancy state $s_\tau^{\mathtt{c},o_\tau^{\pw}}$ associated with $s_\tau$, action probability $a_\tau^{\pw}(u_\tau^{\pw}|o_\tau^{\pw})$, observation $z_{\tau+1}^{\pw}$, and decision rule $a_\tau^{\pb}$, \ie for all $y\in X$ and $o_{\tau+1} = (o_\tau ,u ,z )\in O_{\tau+1}$
\begin{align*}
s_{\tau+1}^{\mathtt{m},\pw}(o_{\tau+1}^{\pw})\cdot s_{\tau+1}^{\mathtt{c}, o_{\tau+1}^{\pw}}(y,o_{\tau+1}^{\pb})  &=\textstyle  a_\tau^{\pw}(u^{\pw}|o_\tau^{\pw}) \cdot a_\tau^{\pb}(u^{\pb}|o_\tau^{\pb})\sum_{x \in X}s_\tau(x,o_\tau)\cdot p_{xy}^{uz}.
\end{align*}
Exploiting this insight along with the linearity of $\alpha$ yields:
\begin{align*} 
\upsilon^{\pb}_{\tau+1}(T(s_\tau,a_\tau)) =\max_{\xi \in \triangle(\mathbb{V})}  \sum_{o_{\tau+1}^{\pw}} \sum_{V \in \mathbb{V}} \min_{\alpha\in V} a_\tau^{\pw}(u^{\pw}|o_\tau^{\pw}) \xi(V) \sum_{o_\tau^{\pb}, u^{\pb}} a_\tau^{\pb}(u^{\pb}|o_\tau^{\pb}) \sum_{x,y,z^{\pb}} \alpha(y,o_{\tau+1}) s_\tau(x,o_\tau) p_{xy}^{uz}.
\end{align*}
Define the following two intermediate functions $g_{V}^{\alpha}$ and $q_{V}^{\alpha}$ for all joint histories $ouz$, 
\begin{align*}
 &\textstyle g_{V}^{\alpha}\colon(ouz^{\pw})\mapsto   \sum_x  s_\tau(x,o) ( \frac{r_{xu}}{|Z^{\pw}|} + \gamma \sum_{y,z^{\pb}}\alpha(y,o^{\pb}u^{\pb}z^{\pb}) p_{xy}^{uz})\\
&\textstyle q_{V}^{\alpha}\colon(o^{\pw}u^{\pw}z^{\pw})\mapsto  \sum_{o^{\pb}}\sum_{u^{\pb}} a^{\pb}(u^{\pb}|o^{\pb}) \cdot g_{V}^{\alpha}(ouz^{\pw}).
\end{align*}
Consequently, the action value $q^{\pb}_\tau(s_\tau,a_\tau)$ can be rewritten as follows:
\begin{align*} 
&\max_{\xi \in \triangle(\mathbb{V})} {\textstyle \sum_{o^{\pw}}\sum_{u^{\pw}} a^{\pw}(u^{\pw}|o^{\pw}) \sum_{V \in \mathbb{V}} \xi(V) \sum_{z^{\pw}}}\min_{\alpha\in V}   q_{V}^{\alpha}(o^{\pw}u^{\pw}z^{\pw}).
\end{align*}
Thus, the greedy decision rule of player ${\pb}$ is the solution of the following optimization problem:
\begin{align*}
\max_{a^{\pb}}\min_{a^{\pw}}\max_{\xi \in \triangle(\mathbb{V})}  \sum_{o^{\pw},u^{\pw}} a^{\pw}(u^{\pw}|o^{\pw}) \sum_{V \in \mathbb{V}} \xi(V) \sum_{z^{\pw}}\min_{\alpha\in V}   q_{V}^{\alpha}(o^{\pw}u^{\pw}z^{\pw}).
\end{align*}
Since the function 
$ o^{\pw} \mapsto  \sum_{u^{\pw}} a^{\pw}(u^{\pw}|o^{\pw}) \sum_{V \in \mathbb{V}} \xi(V) \sum_{z^{\pw}}\min_{\alpha\in V}   q_{V}^{\alpha}(o^{\pw}u^{\pw}z^{\pw})$ 
is a bilinear function of $a^{\pw}$ and $\xi$ as the quantities $\sum_{z^{\pw}}\min_{\alpha\in V}   q_{V}^{\alpha}(o^{\pw}u^{\pw}z^{\pw})$ are constants \wrt any pair $(a^{\pw},\xi)$.
Using Von Neuman minimax theorem \citep{Neumann1928}, one can swap both operators $\min_{a^{\pw}}$ and $\max_{\xi}$  with no loss in optimality, \ie 
\begin{align*}
\max_{\xi,a^{\pb}}\min_{a^{\pw}}  \sum_{o^{\pw},u^{\pw}} a^{\pw}(u^{\pw}|o^{\pw}) \sum_{V \in \mathbb{V}} \xi(V) \sum_{z^{\pw}}\min_{\alpha\in V}   q_{V}^{\alpha}(o^{\pw}u^{\pw}z^{\pw}).
\end{align*}
Let us define the decision variable $\theta(V,u^{\pb}|o^{\pb})  = \xi(V) \cdot a^{\pb}(u^{\pb}|o^{\pb})$ and function $\beta_{V}^{\alpha}\colon (o^{\pw}u^{\pw}z^{\pw}) \mapsto  \xi(V) \cdot q_{V}^{\alpha} (o^{\pw}u^{\pw}z^{\pw}) = \sum_{o^{\pb}}\sum_{u^{\pb}} \theta(V,u^{\pb}|o^{\pb}) \cdot g_{V}^{\alpha}(ouz^{\pw})$, such that the greedy decision rule is the solution of the maximin optimization problem:
\begin{align*}
\max_{\theta}\min_{a^{\pw}} \textstyle \sum_{o^{\pw},u^{\pw}} a^{\pw}(u^{\pw}|o^{\pw}) \sum_{V \in \mathbb{V}}  \sum_{z^{\pw}}\min_{\alpha\in V}   \sum_{o^{\pb}}\sum_{u^{\pb}} \theta(V,u^{\pb}|o^{\pb}) \cdot g_{V}^{\alpha}(ouz^{\pw}).
\end{align*}
Using Wald's maximin model we can convert this maximin optimization problem into a maximization mathematical program, \ie 
\begin{equation*}
\begin{aligned}
\text{maximize}&\textstyle  \sum_{o^{\pw}} \alpha_\theta(o^{\pw}) \\
\text{subject to } &\left|
\begin{array}{l}
 \alpha_\theta(o^{\pw}) \leq \sum_{V}\sum_{z^{\pw}}  \beta_{V}(o^{\pw}u^{\pw}z^{\pw}), ~\forall u^{\pw} \in U^\pw, o^{\pw} \in O^\pw \\
\beta_{V}(o^{\pw}u^{\pw}z^{\pw}) \leq \sum_{o^{\pb}}\sum_{u^{\pb}} \theta(V,u^{\pb}|o^{\pb}) \cdot g_{V}^{\alpha}(ouz^{\pw}),~\forall V \in \mathbb{V}, \alpha \in V, o^{\pw} \in O^\pw, u^{\pw} \in U^\pw, z^{\pw} \in Z^\pw \\
\sum_{V}\sum_{u^{\pb}}  \theta(V,u^{\pb}|o^{\pb}) = 1,~ \forall o^{\pb} \in O^\pb
\end{array}
\right.\\
\text{variables } &\left|
\begin{array}{l}
\theta(V,u^{\pb}|o^{\pb})\in [0,1],~ \forall V \in \mathbb{V}, o^{\pb} \in O^{\pb}, u^{\pb} \in U^\pb
\end{array}
\right.
\end{aligned}   
\end{equation*}
Then, the solution of the linear program in Figure \ref{fig:greedy:lp} is the greedy decision rule of player ${\pb}$, which ends the proof.
\end{proof}

\section{Proof of Theorem \ref{thmsimstomgbellmanoptimalityeqn}}
\label{proofthmsimstomgbellmanoptimalityeqn}
\thmsimstomgbellmanoptimalityeqn*
\begin{proof}
The proof starts from a formal description of a sound point-based update operator. The updated value function $\mathbb{H}^{\pb} (\upsilon^{\pb}_{\tau+1}, s_\tau, \mathbb{G}^{\pb} (\upsilon^{\pb}_{\tau+1}, s_\tau))$, given by $\mathbb{V}_\tau = \mathbb{V}_\tau \cup \{V_{\mathbb{G}^{\pb} (\upsilon^{\pb}_{\tau+1}, s_\tau)}\}$, is sound if the following inequalities hold: 
\begin{enumerate}[($\pmb{A}_1$)]
    \item $\mathbb{H}^{\pb} (\upsilon^{\pb}_{\tau+1}, s_\tau, \mathbb{G}^{\pb} (\upsilon^{\pb}_{\tau+1}, s_\tau))(s_\tau) \geq \upsilon^{\pb}_\tau(s_\tau) $,
    \item $\mathbb{H}^{\pb} (\upsilon^{\pb}_{\tau+1}, s_\tau, \mathbb{G}^{\pb} (\upsilon^{\pb}_{\tau+1}, s_\tau))(\bar{s}_\tau)\leq \mathbb{H}^{\pb} (\upsilon^{\pb}_{\tau+1}, \bar{s}_\tau, \mathbb{G}^{\pb} (\upsilon^{\pb}_{\tau+1}, \bar{s}_\tau))(\bar{s}_\tau)$, for $\bar{s}_\tau \neq s_\tau$.
\end{enumerate}
The proof of inequality ($\pmb{A}_1$) follows that of Theorem \ref{thm:greedy} since $\mathbb{G}^{\pb} (\upsilon^{\pb}_{\tau+1}, s_\tau)$ is the outcome of greedy action-selection operator at occupancy state $s_\tau$. It only remains to prove inequality ($\pmb{A}_2$). To do so, we need to prove the value generated at occupancy state $\bar{s}_\tau\in \bar{S}_\tau$ by the collection of linear functions across conditional occupancy states, $V_{\mathbb{G}^{\pb} (\upsilon^{\pb}_{\tau+1}, s_\tau)}$, constructed for $s_\tau$. 

Figure \ref{fig:greedy:lp:retrival} constructs a linear function for each conditional occupancy state $\bar{s}^{\mathtt{c},o^{\pw}}$ from occupancy state $\bar{s}_\tau\in \bar{S}_\tau$ assuming a fixed policy of player $\pb$. The constructed linear function for conditional occupancy state $\bar{s}^{\mathtt{c},o^{\pw}}$ minimizes the value of player $\pw$ assuming a fixed policy of player $\pb$. Because the fixed policy of $\pb$ was not constructed to maximize value of player $\pb$ at occupancy state $\bar{s}_\tau$, the resulted linear function at $\bar{s}^{\mathtt{c},o^{\pw}}$ lower-bounds the one that would have been created if player $\pb$ maximized upon occupancy state $\bar{s}_\tau$. As a consequence, inequality ($\pmb{A}_2$) holds. Which ends the proof.
\end{proof}

\section{Proof of Theorem \ref{thmerrorbound}}
\label{proofthmerrorbound}
\thmerrorbound*
\begin{proof}
This proof builds upon Theorem \ref{thm:cunha:convex}. Let $s^*_\tau$ an occupancy state, never visited during the planning phase by the PBVI algorithm. There exists an occupancy state $s_\tau \in \bar{S}_\tau$ suth that $\|s^*_\tau - s_\tau \|_1 \leq \delta_{ \bar{S}_{0:}}$. Let $\alpha_\tau$ be a linear function, over occupancy states expressed as distributions over conditional occupancy states, computed during the planning phase by the PBVI algorithm at occupancy state $s_\tau$ after a backup. Let $\alpha^*_\tau$ be the optimal linear function at occupancy state $s^*_\tau$ that would have been computed during the planning phase by the PBVI algorithm at occupancy state $s^*_\tau$ after a backup.
Let $\epsilon_\tau = |\alpha^*_\tau(s^*_\tau) - \alpha_\tau(s_\tau)|$, then 
\begin{align*}
    \epsilon_\tau &= \alpha^*_\tau(s^*_\tau) - \alpha_\tau(s_\tau)\\
                  &= \alpha^*_\tau(s^*_\tau) - \alpha_\tau(s_\tau) + \textcolor{red}{\alpha^*_\tau(s_\tau)- \alpha^*_\tau(s_\tau)},\quad\text{(add zero)}\\
                  &\leq \alpha^*_\tau(s^*_\tau) - \alpha_\tau(s_\tau) + \textcolor{red}{\alpha^*_\tau(s_\tau)}- \alpha_\tau(s_\tau),\quad\text{(since $\alpha^*_\tau(s_\tau)\leq \alpha_\tau(s_\tau)$ )} \\
                  &= (\alpha^*_\tau-\alpha_\tau)\cdot(s^*_\tau-s_\tau),\quad\text{(re-arranging terms)}\\
                  &\leq \|\alpha^*_\tau-\alpha_\tau\|_\infty \|s^*_\tau-s_\tau\|_1,\quad\text{(by Holder's iniequality)}\\
                  &\textstyle \leq c \delta_{ \bar{S}_{0:}} \sum_{\tau=t}^{\ell-1} \gamma^{t-\tau},\quad\text{(since $\|s^*_\tau - s_\tau \|_1 \leq \delta_{ \bar{S}_{0:}}$ and $\|r(\cdot,\cdot)\|_\infty \leq c$)} \\
                  &=c \delta_{ \bar{S}_{0:}} \frac{1-\gamma^{\ell-\tau}}{1-\gamma}.
\end{align*}
 Now, after $\ell$ backups, the error the PBVI algorithm makes is given by $\epsilon  = \sum_{\tau} \gamma^\tau \cdot (\alpha^*_\tau(s^*_\tau) - \alpha_\tau(s_\tau)) = \sum_{\tau} \gamma^\tau \cdot\epsilon_\tau = 2c \delta_{ \bar{S}_{0:}} \frac{1+\ell \gamma^{\ell+1}-(\ell+1)\gamma^\ell}{(1-\gamma)^2}$. Which ends the proof. 
\end{proof}

\section{Pruning Methods}

\subsection{Exact Pruning Method}

The mixed-integer linear program \ref{milp:dominance} tests whether a collection $V_\tau$ is dominated by the collections in $\mathbb{V}_\tau$. This mixed-integer linear program finds the occupancy state $s_\tau$ in which the value function $\upsilon_\tau$ improved the most by added the collection $V_\tau$ to $\mathbb{V}_\tau$. If the value $\delta_{V_\tau}$ maximized by the mixed-integer linear program is non-positive, the collection $V_\tau$ is dominated. Otherwise, the collection $V_\tau$ is not dominated and $\delta_{V_\tau}$ is the amount by which it gives a better value for the occupancy state $s_\tau$ than any collection in $\mathbb{V}_\tau$. While the pruning technique we discussed so far applies over the entire occupancy space, one can drastically reduced time and space complexities by attention of value functions over sample subsets $\bar{S}_{0:} = (\bar{S}_0,\bar{S}_1,\ldots,\bar{S}_{\ell-1})$. A point-based pruning algorithm takes as input a family of collections $\mathbb{V}_\tau$ and returns a subset $\mathbb{V}'_\tau \subseteq \mathbb{V}_\tau$. 

\begin{equation}
\begin{aligned}
\text{maximize}&\textstyle   \quad\delta_{V_\tau} \\
\text{subject to } &\left|
\begin{array}{l}
 \delta_{V_\tau} \leq \sum_{o^{\pw}\in O^{\pw}} (w^{o^{\pw}}_{V_\tau} - w^{o^{\pw}}_{V^{(\kappa)}_\tau}),\quad \forall V^{(\kappa)}_\tau\in \mathbb{V}_\tau \\
 w^{o^{\pw}}_{V_\tau} \leq s_\tau^{\mathtt{m},\pw}(o^{\pw}) \alpha_\tau(s_\tau^{\mathtt{c},o^{\pw}}),\quad \forall \alpha_\tau\in V_\tau  \\
 w^{o^{\pw}}_{V^{(\kappa)}_\tau} \leq s_\tau^{\mathtt{m},\pw}(o^{\pw}) \alpha_\tau(s_\tau^{\mathtt{c},o^{\pw}}),\quad \forall V^{(\kappa)}_\tau\in \mathbb{V}_\tau, \forall\alpha_\tau\in V^{(\kappa)}_\tau  \\
 w^{o^{\pw}}_{V^{(\kappa)}_\tau} \geq s_\tau^{\mathtt{m},\pw}(o^{\pw}) \alpha^{(\kappa)}_\tau(s_\tau^{\mathtt{c},o^{\pw}}) +  (1-c^{o^{\pw}}_{\alpha^{(\kappa)}_\tau})\pmb{\mathbb{M}},\quad \forall V^{(\kappa)}_\tau\in \mathbb{V}_\tau, \forall \alpha^{(\kappa)}_\tau \in V^{(\kappa)}_\tau  \\
\sum_{V^{(\kappa)}_\tau \in \mathbb{V}_\tau} \sum_{\alpha^{(\kappa)}_\tau \in V^{(\kappa)}_\tau} c^{o^{\pw}}_{\alpha^{(\kappa)}_\tau} = 1\\
\sum_{x\in X}\sum_{o\in O_\tau} s_\tau(x,o) = 1
\end{array}
\right.\\
\text{variables } &\left|
\begin{array}{l}
c^{o^{\pw}}_{\alpha^{(\kappa)}_\tau} \in \{0,1\}, \forall V^{(\kappa)}_\tau\in \mathbb{V}_\tau, \forall \alpha^{(\kappa)}_\tau \in V^{(\kappa)}_\tau \\
s_\tau(x,o) \in [0,1]\\
\delta_{V_\tau}\in \mathbb{R}.
\end{array}
\right.
\end{aligned}
\label{milp:dominance}
\end{equation}

\subsection{Bounded Pruning Method}

The bounded pruning, \cf Algorithm \ref{pruning:zsposg}, is a pruning algorithm that works under the assumption that the subsets $\bar{S}_{0:}$ of occupancy states are finite. The algorithm is bounded in the sense that the families $\mathbb{V}_{0:}$ of collections of linear functions are guaranteed to be no larger than $|\bar{S}_{0:}|$. It tests whether a collection $V_\tau$ is dominated by the collections in $\mathbb{V}_\tau$ if for any occupancy state $s_\tau\in \bar{S}_\tau$,
$$\max_{V'_\tau\in \mathbb{V}_\tau} \sum_{o^{\pw}\in O^{\pw}} s^{\mathtt{m},\pw}(o^{\pw}) \min_{\alpha\in V'_\tau} \alpha(s^{\mathtt{c},o^{\pw}}) -  \sum_{o^{\pw}\in O^{\pw}} s^{\mathtt{m},\pw}(o^{\pw}) \min_{\alpha\in V_\tau} \alpha(s^{\mathtt{c},o^{\pw}}) \geq 0.$$
Its result $\mathbb{V}'_\tau$ is typically smaller than the minimum-size subset that is sound with
respect to the whole occupancy simplex, \eg resulting from the mixed-integer linear program \ref{milp:dominance}.
With catching, the inner products, \eg $\alpha_\tau(s_\tau^{\mathtt{c},o^{\pw}})$, need to be computed only once for each collection $V_\tau\in \mathbb{V}_\tau$, each vector $\alpha_\tau\in V_\tau$ and each conditional occupancy state $s_\tau^{\mathtt{c},o^{\pw}}$, which makes the running time of the bounded pruning method about $\pmb{O}(|\mathbb{V}_\tau||\bar{V}_\tau||\bar{S}_\tau||X||\bar{O}_\tau(s_\tau)|)$ where $\bar{O}_\tau(s_\tau) = \max_{s_\tau\in \bar{S}_\tau}\{o| o\in O_\tau, \Pr\{o|s_\tau\}>0  \}$. 

        \begin{algorithm}[H]
        \caption{Bounded pruning for $M'$.}
        \label{pruning:zsposg}
        \begin{algorithmic}
            \STATE ${\mathtt{function}~ \mathtt{Bounded\text{-}Pruning}(\mathbb{V},\bar{S}, \delta)}$
            \FOR{$V\in \mathbb{V}$}
                \STATE $\mathtt{refCount}(V) \gets 0$.
            \ENDFOR
            \FOR{$s\in \bar{S}$}
                \STATE $$V^* \gets \argmax_{V\in \mathbb{V}} \sum_{o^{\pw}\in O^{\pw}} s^{\mathtt{m},\pw}(o^{\pw}) \min_{\alpha\in V} \alpha(s^{\mathtt{c},o^{\pw}})$$
                \STATE increment $\mathtt{refCount}(V^*)$
            \ENDFOR
            \STATE \textbf{return} $\{V\in \mathbb{V} \mid \mathtt{refCount}(V)>0\}$
        \end{algorithmic}
        \end{algorithm}

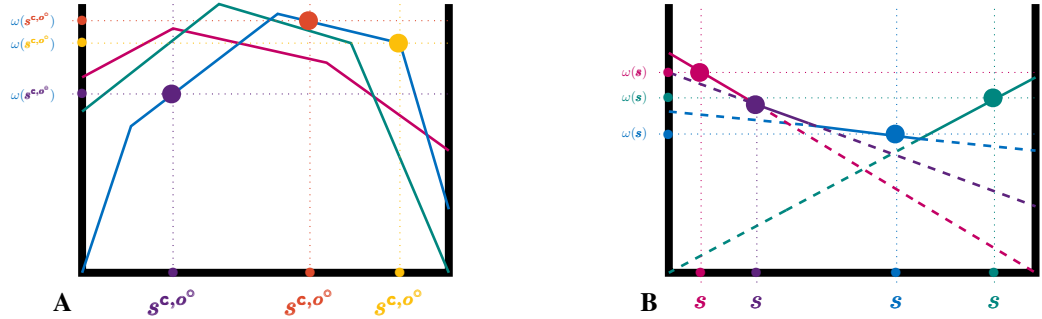
\begin{figure}[H]
\centering
\begin{tikzpicture}[
        scale=1.3,
        IS/.style={sthlmRed, thick},
        LM/.style={sthlmRed, thick},
        axis/.style={very thick, ->, >=stealth', line join=miter},
        important line/.style={thick}, dashed line/.style={dashed, thin},
        every node/.style={color=black},
        dot/.style={circle,fill=red,minimum size=8pt,inner sep=0pt, outer sep=-1pt},
    ]

    \draw[axis, -, line width=.3em] (-6,2.75) -- (-6,0) -- (-2.25,0) -- (-2.25,2.75);
    \node at (-6.2,-0.3) {\textbf{A}};
    \node at (-0.2,-0.3) {\textbf{B}};

    \draw[-, sthlmRed, line width=.1em] (-6,2) -- (-5.07, 2.5) -- (-3.5, 2.15) -- (-2.25, 1.25);

    \draw[-, sthlmGreen, line width=.1em] (-6,1.65) -- (-4.6, 2.75) -- (-3.25, 2.35) -- (-2.25,0);

    \draw[-, sthlmBlue, line width=.1em] (-6,0)-- (-5.5,1.5) -- (-4,2.65) -- (-2.75,2.35) -- (-2.25, .65);

    \node[scale=1, sthlmYellow] at (-2.75,-0.3) {$\pmb{s^{\mathtt{c},o^{\circ}}}$};
    
    \node[scale=1, sthlmPurple] at (-5.07,-0.3) {$\pmb{s^{\mathtt{c},o^{\circ}}}$};

    \node[scale=1, sthlmOrange] at (-3.67,-0.3) {$\pmb{s^{\mathtt{c},o^{\circ}}}$};

    \draw[-, dotted, sthlmYellow] (-2.75,2.75) -- (-2.75,0);
    \draw[-, dotted, sthlmYellow] (-6.1,2.35) -- (-2.25,2.35);    
    \node[rotate=90, scale=2, sthlmYellow] at (-2.75,2.35) {$\bullet$};
    \node[scale=.5, sthlmBlue] at (-6.5,2.35) {$\omega(\textcolor{sthlmYellow}{\pmb{s^{\mathtt{c},o^{\circ}}}})$};    
    \node[scale=1, sthlmYellow] at (-2.75,0) {$\bullet$};
    \node[scale=1, sthlmYellow] at (-6,2.35) {$\bullet$};

    \draw[-, dotted, sthlmOrange] (-3.67,2.75) -- (-3.67,0);
    \draw[-, dotted, sthlmOrange] (-6.1,2.58) -- (-2.25,2.58);    
    \node[rotate=90, scale=2, sthlmOrange] at (-3.67,2.58) {$\bullet$};
    \node[scale=.5, sthlmBlue] at (-6.5,2.58) {$\omega(\textcolor{sthlmOrange}{\pmb{s^{\mathtt{c},o^{\circ}}}})$};    
    \node[scale=1, sthlmOrange] at (-3.67,0) {$\bullet$};
    \node[scale=1, sthlmOrange] at (-6,2.58) {$\bullet$};

    \draw[-, dotted, sthlmPurple] (-5.07,2.75) -- (-5.07,0);
    \draw[-, dotted, sthlmPurple] (-6.1,1.83) -- (-2.25,1.83);    
    \node[rotate=90, scale=2, sthlmPurple] at (-5.07,1.83) {$\bullet$};
    \node[scale=.5, sthlmBlue] at (-6.5,1.83) {$\omega(\textcolor{sthlmPurple}{\pmb{s^{\mathtt{c},o^{\circ}}}})$};    
    \node[scale=1, sthlmPurple] at (-5.07,0) {$\bullet$};
    \node[scale=1, sthlmPurple] at (-6,1.83) {$\bullet$};

    \draw[axis, -, line width=.3em] (0,2.75) -- (0,0) -- (3.75,0) -- (3.75,2.75);

    \draw[-, sthlmGreen, line width=.1em] (3.75,2) -- (2.55,1.36015);
    \draw[-, sthlmGreen, line width=.1em, dashed] (2.55,1.36015) -- (1.2,0.6406);
    \draw[-, sthlmGreen, line width=.1em, dashed] (1.2,0.6406) -- (0,0);

    \draw[-, sthlmBlue, line width=.1em, dashed] (0,1.65)--(1.43, 1.51);     
    \draw[-, sthlmBlue, line width=.1em] (1.43, 1.51)--(2.57,1.37);     
    \draw[-, sthlmBlue, line width=.1em, dashed] (2.57,1.37) -- (3.75,1.25);     

     \draw[-, sthlmRed, line width=.1em] (0,2.25)--(0.9, 1.72); 
    \draw[-, sthlmRed, line width=.1em, dashed] (0.9, 1.72)--(2.55,0.72); 
    \draw[-, sthlmRed, line width=.1em, dashed] (2.55,0.72)--(3.75,0); 

     \draw[-, sthlmPurple, line width=.1em, dashed] (0, 2.05)--(0.9, 1.72); 
    \draw[-, sthlmPurple, line width=.1em] (0.9, 1.72)--(1.43, 1.53); 
    \draw[-, sthlmPurple, line width=.1em, dashed] (1.43, 1.53)--(3.75, 0.68);

    \draw[-, dotted, sthlmBlue] (2.33,-0.1) -- (2.33,2.75);    
    \draw[-, dotted, sthlmBlue] (-.1,1.42) -- (3.75,1.42);    
    \node[scale=1, sthlmBlue] at (2.33,-0.3) {$\pmb{s}$};    
    \node[scale=.5, sthlmBlue] at (-.35,1.42) {$\omega(\textcolor{sthlmBlue}{\pmb{s}})$};    
    \node[rotate=90, scale=2, sthlmBlue] at (2.33,1.42) {$\bullet$};
    \node[rotate=90, scale=1, sthlmBlue] at (2.33,0) {$\bullet$};
    \node[rotate=90, scale=1, sthlmBlue] at (0,1.42) {$\bullet$};

    \draw[-, dotted, sthlmRed] (.33,-0.1) -- (.33,2.75);    
    \draw[-, dotted, sthlmRed] (-.1,2.05) -- (3.75,2.05);    
    \node[scale=1, sthlmRed] at (.33,-0.3) {$\pmb{s}$};    
    \node[scale=.5, sthlmRed] at (-.35,2.05) {$\omega(\textcolor{sthlmRed}{\pmb{s}})$};    
    \node[rotate=90, scale=2, sthlmRed] at (.33,2.05) {$\bullet$};
    \node[rotate=90, scale=1, sthlmRed] at (.33,0) {$\bullet$};
    \node[rotate=90, scale=1, sthlmRed] at (0,2.05) {$\bullet$};

    \draw[-, dotted, sthlmGreen] (3.33,-0.1) -- (3.33,2.75);    
    \draw[-, dotted, sthlmGreen] (-.1,1.79) -- (3.75,1.79);    
    \node[scale=1, sthlmGreen] at (3.33,-0.3) {$\pmb{s}$};    
    \node[scale=.5, sthlmGreen] at (-.35,1.79) {$\omega(\textcolor{sthlmGreen}{\pmb{s}})$};    
    \node[rotate=90, scale=2, sthlmGreen] at (3.33,1.79) {$\bullet$};
    \node[rotate=90, scale=1, sthlmGreen] at (3.33,0) {$\bullet$};
    \node[rotate=90, scale=1, sthlmGreen] at (0,1.79) {$\bullet$};

    \draw[-, dotted, sthlmPurple] (0.9, -0.1) -- (0.9, 1.72);     
    \node[scale=1, sthlmPurple] at (0.9, -0.3) {$\pmb{s}$};
    \node[rotate=90, scale=2, sthlmPurple] at (0.9, 1.72) {$\bullet$};
    \node[rotate=90, scale=1, sthlmPurple] at (0.9, 0) {$\bullet$};

\end{tikzpicture}
\caption{Comparison between exact and bounded pruning methods.}
\label{figure:convex:function:representations:pruning}
\end{figure}  

Figure \ref{figure:convex:function:representations:pruning} compares the results of the exact and bounded pruning approaches on an example family of collections of linear functions. (Right) The purple collection would be pruned under bounded pruning, the foundation point of the purple line being on the red line. However, it would not be pruned under exact pruning, because it is strictly optimal for occupancy points outside of its foundation point (purple point). \textbf{Best viewed in color.}


\section{Additional Plots}
\label{sec:appendix:additional:plots}
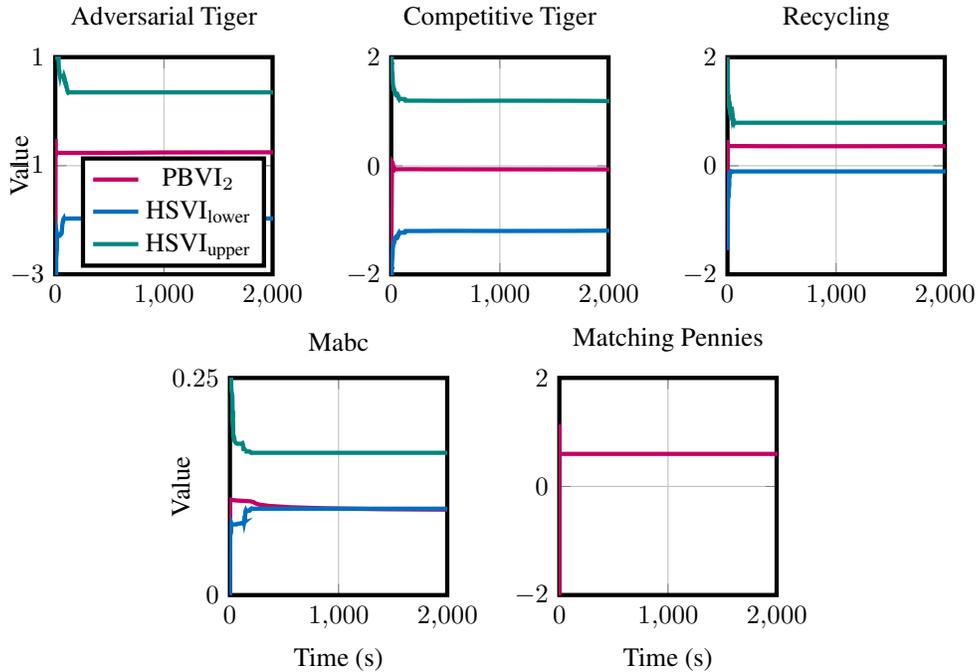
\begin{figure}[H]
  \centering
  \begin{subfigure}{0.32\textwidth}
    \centering
    \begin{tikzpicture}
      \begin{axis}[
        title={Adversarial Tiger},
        width=\linewidth,
        height=\linewidth,
        ylabel={Value},
        ylabel style={yshift=-5mm},
        xmin=0, xmax=2000,
        xtick={0, 1000, 2000}, 
        ymin=-3, ymax=1,
        ytick={-3, -1, 1},
        grid=both,
        style={ultra thick},
        legend pos=south east,
        ]
        \addplot+[smooth,mark=none,sthlmRed] table [col sep=comma, x=time, y=value] {./paper/results/PBVI_pruning_hyperplane/H4/adversarialtiger.csv};
        \addplot+[smooth,mark=none,sthlmBlue] table [col sep=comma, x=time, y=lb] {./paper/results/HSVI_Delage/H4/adversarialtiger.csv};
        \addplot+[smooth,mark=none,sthlmGreen] table [col sep=comma, x=time, y=ub] {./paper/results/HSVI_Delage/H4/adversarialtiger.csv};
        \legend{PBVI$_2$, HSVI$_\text{lower}$, HSVI$_\text{upper}$}
      \end{axis}
    \end{tikzpicture}
    \label{fig:lineplots:hsvi_vs_pbvi:adversarial_tiger}
  \end{subfigure}%
  \begin{subfigure}{0.32\textwidth}
    \centering
    \begin{tikzpicture}
      \begin{axis}[
        title={Competitive Tiger},
        width=\linewidth,
        height=\linewidth,
        xmin=0, xmax=2000,
        xtick={0, 1000, 2000}, 
        ymin=-2, ymax=2,
        ytick={-2, 0, 2},
        grid=major,
        style={ultra thick},
        ]
        \addplot+[smooth,mark=none,sthlmRed] table [col sep=comma, x=time, y=value] {./paper/results/PBVI_pruning_hyperplane/H4/competitivetiger.csv};
        \addplot+[smooth,mark=none,sthlmBlue] table [col sep=comma, x=time, y=lb] {./paper/results/HSVI_Delage/H4/competitivetiger.csv};
        \addplot+[smooth,mark=none,sthlmGreen] table [col sep=comma, x=time, y=ub] {./paper/results/HSVI_Delage/H4/competitivetiger.csv};
      \end{axis}
    \end{tikzpicture}
    \label{fig:lineplots:hsvi_vs_pbvi:competitive_tiger}
  \end{subfigure}%
  \begin{subfigure}{0.32\textwidth}
    \centering
    \begin{tikzpicture}
      \begin{axis}[
        title={Recycling},
        width=\linewidth,
        height=\linewidth,
        ylabel style={yshift=-5mm},
        xmin=0, xmax=2000,
        xtick={0, 1000, 2000}, 
        ymin=-2, ymax=2,
        ytick={-2, 0, 2},
        grid=both,
        style={ultra thick},
        ]
        \addplot+[smooth,mark=none,sthlmRed] table [col sep=comma, x=time, y=value] {./paper/results/PBVI_pruning_hyperplane/H4/recycling.csv};
        \addplot+[smooth,mark=none,sthlmBlue] table [col sep=comma, x=time, y=lb] {./paper/results/HSVI_Delage/H4/recycling.csv};
        \addplot+[smooth,mark=none,sthlmGreen] table [col sep=comma, x=time, y=ub] {./paper/results/HSVI_Delage/H4/recycling.csv};
        \end{axis}
    \end{tikzpicture}
    \label{fig:lineplots:hsvi_vs_pbvi:recycling}
  \end{subfigure}
  \begin{subfigure}{0.32\textwidth}
    \centering
    \begin{tikzpicture}
      \begin{axis}[
        title={Mabc},
        width=\linewidth,
        height=\linewidth,        
        xlabel={Time (s)},
        ylabel={Value},
        ylabel style={yshift=-5mm},
        xmin=0,
        xmax=2000,
        xtick={0, 1000, 2000}, 
        ymin=0.,
        ymax=0.25,
        ytick={0., 0.25},
        grid=both,
        style={ultra thick}
        ]
        \addplot+[smooth, mark=none, sthlmRed] table [col sep=comma, x=time, y=value] {./paper/results/PBVI_pruning_hyperplane/H4/mabc.csv};
        \addplot+[smooth, mark=none, sthlmBlue] table [col sep=comma, x=time, y=lb] {./paper/results/HSVI_Delage/H4/mabc.csv};
        \addplot+[smooth, mark=none, sthlmGreen] table [col sep=comma, x=time, y=ub] {./paper/results/HSVI_Delage/H4/mabc.csv};
        \end{axis}
    \end{tikzpicture}
    \label{fig:lineplots:hsvi_vs_pbvi:mabc}
  \end{subfigure}%
  \begin{subfigure}{0.32\textwidth}
    \centering
    \begin{tikzpicture}
      \begin{axis}[
        title={Matching Pennies},
        width=\linewidth,
        height=\linewidth,
        xlabel=Time (s),
        xmin=0, xmax=2000,
        xtick={0, 1000, 2000}, 
        ymin=-2, ymax=2,
        ytick={-2, 0, 2},
        grid=both,
        style={ultra thick}
        ]
        \addplot+[smooth,mark=none,sthlmRed] table [col sep=comma, x=time, y=value] {./paper/results/PBVI_pruning_hyperplane/H4/matchingpennies.csv};
        \end{axis}
    \end{tikzpicture}
    \label{fig:lineplots:hsvi_vs_pbvi:matching_pennies}
  \end{subfigure}%
  \caption{A visual representation of the performance of our best performing algorithm (PBVI$_2$) against the HSVI algorithm of \citet{DelBufDibSaf-DGAA-23} for horizon $\ell = 4$ across five different games. \textbf{Best viewed in color.}}
  \label{fig:hsvi_vs_pbvi_results}
\end{figure}

\section{Experiments on Larger Benchmarks}
\label{sec:appendix:largerbenchmark}
Preliminary experiments on larger benchmarks were already carried out to test the scalability of the method. Both chosen benchmarks exhibit a larger hidden state, action and observation space, yet PBVI$_3$ seems to show very promising results.

\begin{table}[H]
\centering
\caption{Snapshot of empirical results for larger benchmarks (Box pushing and Mars). For each game, we report time (in hours) and the best value for horizons $\ell$ computed with PBVI$_3$. {\sc oot} means a time limit of 2 hours has been exceeded.}
\label{table:larger_benchmark_results}
\begin{tabular}{@{}ccrr@{}} 
\toprule
Problem & $\ell$ & \multicolumn{2}{c}{{\bfseries PBVI$_3$}} \\ 
\cmidrule{1-2} \cmidrule{3-4}
\multirow{6}{*}{\shortstack[C]{Box\\pushing}} & 2 & 0 & 4.1 \\ 
& 3 & 2.36 & 4.92 \\ 
& 4 & OOT & 123 \\ 
& 5 & OOT & 216 \\ 
& 7 & OOT & 400 \\ 
& 10 & OOT & 723 \\ 
\cmidrule{1-2} \cmidrule{3-4}
\multirow{2}{*}{Mars} & 2 & 0.08 & -1.12 \\ 
& 3 & 40 & -4.0 \\ 
\bottomrule
\end{tabular}
\end{table}

\section{Experiments on Subclasses}
\label{sec:appendix:subclasses}

\begin{table}[H]
\caption{Snapshot of empirical results. For each subclass of game (Markov Game and Public Observable Markov Game) and algorithms, we report time (in seconds) and the best value for the horizon $\ell$.}
\label{table:subclass_results}
\centering
\scriptsize
\begin{tabular}{@{}c c rr | c rr @{}}
\toprule
Problem & \multicolumn{3}{c|}{MG} & \multicolumn{3}{c}{POMG} \\
\cmidrule[0.4pt](lr{0.125em}){1-4} \cmidrule[0.4pt](lr{0.125em}){5-7}
\centering
     & $\ell$ & Time & Value & $\ell$ & Time & Value \\
\midrule
Adversarial tiger 
     & 1 & 0 & 0.25 & 1 & 0 & 0.25 \\
     & 50 & 0 & 11.9 & 2 & 0 & 0 \\
     & 100 & 0 & 23.8 & 3 & 0 & 0.2 \\
     & 1000 & 2 & 238 & 4 & 2 & 0.2 \\
\midrule
Competitive tiger
     & 1 & 0 & 0 & 1 & 0 & 0 \\
     & 50 & 0 & 0 & 2 & 0 & 0 \\
     & 100 & 0 & 0 & 3 & 2 & 0 \\
     & 1000 & 3 & 0 & 4 & 0 & 0 \\
\midrule
Recycling 
     & 1 & 0 & -0.4 & 1 & 0 & -0.4 \\
     & 50 & 0 & 1.36 & 2 & 0 & 0.28 \\
     & 100 & 0 & 2.41 & 3 & 2 & 0.34 \\
     & 1000 & 4 & 21 & 4 & 26 & 0.38 \\
\midrule
Mabc
     & 1 & 0 & 0.05 & 1 & 0 & 0.05 \\
     & 50 & 0 & 0.52 & 2 & 0 & 0.88 \\
     & 100 & 0 & 0.98 & 3 & 0 & 0.99 \\
     & 1000 & 3 & 9.15 & 4 & 4 & 0.11 \\
\midrule
Matching pennies
     & 1 & 0 & 0 & 1 & 0 & 0 \\
     & 50 & 0 & -49 & 2 & 0 & -1 \\
     & 100 & 0 & -99 & 3 & 0 & -2 \\
     & 1000 & 2 & -999 & 4 & 0 & -3 \\
\bottomrule
\end{tabular}
\end{table}

\section{Computational Resources}
\label{sec:computational_resources}
The results presented in this paper were run on CPU. For CFR+, AMD Ryzen was used with 90G of available space. For HSVI/PBVI, a CPU was also used with 30-70G of available space for smaller horizons and 90-110G for higher horizons.

\end{document}